\documentclass[runningheads]{llncs}

\usepackage[USenglish]{babel}
\usepackage{amsfonts,amsmath,mathtools,amssymb}
\usepackage{tikz}
\usepackage{url}
\usetikzlibrary{automata,arrows,positioning}
\usepackage{graphicx}
\usepackage{csquotes}
\usepackage{marvosym}

\makeatletter
\RequirePackage[bookmarks,colorlinks=true]{hyperref}%
   \def\@citecolor{blue}%
   \def\@urlcolor{blue}%
   \def\@linkcolor{blue}%

\def\orcidID#1{\smash{\href{http://orcid.org/#1}{\protect\raisebox{-1.25pt}{\protect\includegraphics{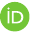}}}}}
\makeatother

\usepackage{cleveref}

\newcommand{\acc}{\mbox{\sc acc}}
\newcommand{\rej}{\mbox{\sc rej}}
\newcommand{\init}{\iota}

\newcommand{\infi}{{\it inf}}
\newcommand{\comp}{{\it comp}}
\newcommand{\zug}[1]{\langle #1 \rangle}
\newcommand{\lang}{L}

\newcommand{\A}{{\cal A}}
\newcommand{\R}{{\cal R}}

\newcommand{\T}{{\cal T}}
\newcommand{\U}{{\cal U}}

\DeclareMathOperator{\DBWreal}{DBW}
\DeclareMathOperator{\NoDBWreal}{NoDBW}
\DeclareMathOperator{\DBWsep}{SepDBW}
\DeclareMathOperator{\NoDBWsep}{NoSepDBW}

\newcommand{\Real}{\mathrm{Real}}
\newcommand{\NoReal}{\mathrm{NoReal}}
\newcommand{\Sep}{\mathrm{Sep}}
\newcommand{\NoSep}{\mathrm{NoSep}}
\newcommand{\Lacc}{\ensuremath{L_{\textnormal{acc}}}}
\newcommand{\Llegal}{\ensuremath{L_{\textnormal{struct}}}}
\newcommand{\dgamw}{D$\gamma$W}

\newif\ifarxiv
\arxivtrue

\ifarxiv
\fi

\begin{document}
\title{Certifying Inexpressibility\thanks{\ifarxiv This is the full version of an article with the same title that appears in the FoSSaCS 2021 conference proceedings.\fi{} Orna Kupferman is supported in part by the Israel Science Foundation, grant No. 2357/19. Salomon Sickert is supported in part by the Deutsche Forschungsgemeinschaft (DFG) under project numbers 436811179 and 317422601 (\enquote{Verified Model Checkers}), and in part funded by the European Research Council (ERC) under the European Union's Horizon 2020 research and innovation programme under grant agreement No. 787367 (PaVeS).}}
%
%

\author{Orna Kupferman\inst{1} \orcidID{0000-0003-4699-6117} \and
Salomon Sickert\inst{1,2} (\Letter) \orcidID{0000-0002-0280-8981}}
\institute{School of Computer Science and Engineering, \\ The Hebrew University, Jerusalem, Israel. \\
\email{orna@cs.huji.ac.il}, \email{salomon.sickert@mail.huji.ac.il}
\and Technische Universit\"at M\"unchen, Munich, Germany. \\
\email{s.sickert@tum.de}}

\authorrunning{Orna Kupferman and Salomon Sickert}
\maketitle              
\vspace{-1.5em}
\begin{abstract}
Different classes of automata on infinite words have different expressive power. Deciding whether a given language $L \subseteq \Sigma^\omega$ can be expressed by an automaton of a desired class can be reduced to deciding a game between Prover and Refuter: in each turn of the game, Refuter provides a letter in $\Sigma$, and Prover responds with an annotation of the current state of the run (for example, in the case of B\"uchi automata, whether the state is accepting or rejecting, and in the case of parity automata, what the color of the state is). Prover wins if the sequence of annotations she generates is correct: it is an accepting run iff the word generated by Refuter is in $L$.
We show how a winning strategy for Refuter can serve as a simple and easy-to-understand certificate to inexpressibility, and how it induces additional forms of certificates.
Our framework handles all classes of deterministic automata, including ones with structural restrictions like weak automata. In addition, it can be used for refuting {\em separation\/} of two languages by an automaton of the desired class, and for finding automata that {\em approximate\/} $L$  and belong to the desired class.
\end{abstract}

\section{Introduction}
Finite {\em automata on infinite objects\/} were first introduced in the 60's, and were the key to the solution of several fundamental decision problems in mathematics and logic \cite{Buc62,McN66,Rab69}. Today, automata on infinite objects are used for specification, verification, and synthesis of nonterminating systems. The automata-theoretic approach reduces questions about systems and their specifications to questions about automata \cite{Kur94,VW94}, and is at the heart of many algorithms and tools.
Industrial-strength property-specification languages such as the IEEE 1850 Standard for Property Specification Language (PSL) \cite{EF06} include regular expressions and/or automata, making specification and verification tools that are based on automata even more essential and popular.

A run $r$ of an automaton on infinite words is an infinite sequence of states, and acceptance is determined with respect to the set of states that $r$ visits infinitely often. For example, in {\em B\"uchi\/} automata, some of the states are designated as accepting states, denoted by $\alpha$, and a run is accepting iff it visits states from the accepting set $\alpha$ infinitely often \cite{Buc62}.
Dually, in {\em co-B\"uchi\/} automata, a run is accepting if it visits the set $\alpha$ only finitely often.
Then, in {\em parity\/} automata, the acceptance condition maps each state  to a color in some set $C=\{j,\ldots,k\}$,  for $j \in \{0,1\}$ and some {\em index\/} $k \geq 0$, and a run is accepting if the maximal color it visits infinitely often is odd.

The different classes of automata have different {\em expressive power}. For example, while deterministic parity automata can recognize all $\omega$-regular languages, deterministic B\"uchi automata cannot  \cite{Lan69}.
We use DBW, DCW, and DPW to denote a deterministic B\"uchi, co-B\"uchi, and parity word automaton, respectively, or (this would be clear from the context) the set of languages recognizable by the automata in the corresponding class.
There has been extensive research on expressiveness of automata on infinite words \cite{Tho90,Kup18}. In particular, researchers have studied two natural expressiveness hierarchies induced by different classes of deterministic automata. The first hierarchy is the {\em Mostowski Hierarchy}, induced by the index of parity automata \cite{Mos84,Wag79}.
Formally, let DPW[$0,k$] denote a DPW with $C=\{0,\ldots,k\}$, and similarly for DPW[$1,k$] and $C=\{1,\ldots,k\}$.  Clearly, DPW[$0,k$]~$\subseteq$~DPW[$0,k+1$], and similarly  DPW[$1,k$]~$\subseteq$~DPW[$1,k+1$]. The hierarchy is infinite and strict. Moreover, DPW[$0,k$] complements DPW[$1,k+1$], and for every $k \geq 0$, there are languages $L_k$ and $L'_k$ such that $L_k \in$~DPW[$0,k$] $\setminus$~DPW[$1,k+1$] and $L'_k \in$~DPW[$1,k+1$]~$\setminus$~DPW[$0,k$]. At the bottom of this hierarchy, we have DBW and DCW. Indeed, DBW=DPW[$0,1$] and DCW=DPW[$1,2$].

While the Mostowski Hierarchy refines DPWs, the
second hierarchy, which we term the {\em depth hierarchy}, refines deterministic {\em weak\/} automata (DWWs). Weak automata can be
viewed as a special case of B\"uchi or co-B\"uchi automata in which
every strongly connected component in the graph induced by the
structure of the automaton is either contained in $\alpha$ or is disjoint from $\alpha$, where $\alpha$ is depending on the acceptance condition the set of accepting or rejecting states.
The structure of weak automata captures the alternation between greatest and least fixed points in many
temporal logics, and they were introduced in this context in \cite{MSS86}. DWWs have
been used to represent vectors of real numbers \cite{BJW01}, and they have many appealing theoretical and practical properties \cite{Lod01,KMM06}.
In terms of expressive power, DWW~=~DCW~$\cap$~DBW.

The depth hierarchy is induced by the depth of alternation between accepting and rejecting components in DWWs. For this, we view a DWW as a DPW in which the colors visited along a run can only increase.
Accordingly, each run eventually gets trapped in a single color, and is accepting iff this color is odd. We use DWW[$0,k$] and DWW[$1,k$] to denote weak-DPW[$0,k$] and weak-DPW[$1,k$], respectively. The picture obtained for the depth hierarchy is identical to that of the Mostowski hierarchy, with DWW[$j,k$] replacing DPW[$j,k$] \cite{Wag79}.
At the bottom of the depth hierarchy we have {\em co-safety\/}  and {\em safety\/} languages \cite{AS87}.
Indeed, co-safety languages are DWW[$0,1$] and safety are DWW[$1,2$].

Beyond the theoretical interest in expressiveness hierarchies, their study is motivated by the fact many algorithms, like synthesis and probabilistic model checking, need to operate on deterministic automata \cite{BCJ18,BAFK18}. The lower the automata are in the expressiveness hierarchy, the simpler are algorithms for reasoning about them. Simplicity goes beyond complexity, which typically depends on the parity index \cite{EJ91}, and involves important practical considerations like minimization and canonicity (exists only for DWWs \cite{Lod01}), circumvention of Safra's determinization \cite{KV05c}, and symbolic implementations \cite{SMV18}. Of special interest is the characterization of DBWs. For example, it is shown in \cite{KV05b} that given a {\em linear temporal logic\/} formula $\psi$, there is an {\em alternation-free $\mu$-calculus\/} formula equivalent to $\forall\psi$ iff $\psi$ can be recognized by a DBW. Further research studies {\em typeness\/} for deterministic automata, examining the ability to define a weaker acceptance condition on top of a given automaton \cite{KPB94,KMM06}.

Our goal in this paper is to provide a simple and easy-to-understand explanation to inexpressibility results.
The need to accompany results of decision procedures by an explanation (often termed ``certificate'') is not new, and includes certification of a ``correct'' decision of a model checker   \cite{KV05d,ACOW18}, reachability certificates in complex multi-agent systems \cite{AL20}, and explainable reactive synthesis \cite{BFT20}. To the best of our knowledge, our work is the first to provide certification to inexpressibility results.

The underlying idea is simple: Consider a language $L$ and a class $\gamma$ of deterministic automata. We consider a turn-based two-player game in which one player (Refuter) provides letters in $\Sigma$, and the second player (Prover) responds with letters from a set $A$ of annotations that describe states in a deterministic automaton. For example, when we consider a DBW, then $A=\{\acc,\rej\}$, and when we consider a DPW[$0,k$], then $A=\{0,\ldots,k\}$. Thus, during the interaction, Refuter generates a word $x \in \Sigma^\omega$ and Prover responds with a word $y \in A^\omega$. Prover wins if for all words $x \in \Sigma^\omega$, we have that $x \in L$ iff $y$ is accepting according to $\gamma$.  Clearly, if there is a deterministic $\gamma$ automaton for $L$, then Prover can win by following its run on $x$. Dually, a finite-state winning strategy for Prover induces a deterministic $\gamma$ automaton for $L$. The game-based approach is not new, and has been used for deciding the membership of given $\omega$-regular languages in different classes of deterministic automata
\cite{KV05c}. Further, the game-based formulation is used in descriptive set theory to classify sets into hierarchies, see for example \cite[Chapters 4 and 5]{PP04} for an introduction that focuses on $\omega$-regular languages. Our contribution is a study of strategies for Refuter. Indeed, since the above described game is determined \cite{BL69} and the strategies are finite-state, Refuter has a winning strategy iff no deterministic $\gamma$ automaton for $L$ exists, and this winning strategy can serve as a certificate for inexpressibility.

\begin{figure}[t]
\begin{center}
\includegraphics{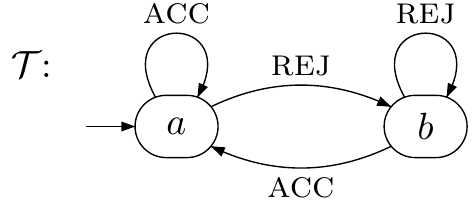}
\caption{A refuter for DBW-recognizability of ``only finitely many $a$'s''.}
\label{ref fg1}
\vspace{-1.5em}
\end{center}
\end{figure}

\begin{example}
Consider the language $L_{\neg \infty a} \subseteq \{a,b\}^\omega$ of all words with only finitely many $a$'s. It is well known that $L$ cannot be recognized by a DBW \cite{Lan69}. In \Cref{ref fg1} we describe what we believe to be the neatest proof of this fact. The figure describes a transducer $\R$ with inputs in $\{\mbox{\sc acc,rej}\}$ and outputs in $\{a,b\}$ -- the winning strategy of Refuter in the above described game. The way to interpret $\R$ is as follows.
In each round of the game, Prover tells Refuter whether the run of her DBW for $L_{\neg \infty a}$ is in an accepting or a rejecting state, and Refuter uses $\R$ in order to  respond with the next letter in the input word. For example, if Prover starts with $\acc$, namely declaring that the initial state of her DBW is accepting, then Refuter responds with $a$, and if Prover continues with $\rej$, namely declaring that the state reachable with $a$ is rejecting, then Refuter responds with $b$. If Prover continues with $\rej$ forever, then Prover continues with $b$ forever. Thus, together Prover and Refuter generate two words: $y \in \{\mbox{\sc acc,rej}\}^\omega$ and $x \in \{a,b\}^\omega$. Prover wins whenever $x \in L_{\neg \infty a}$ iff $y$ contains infinitely many {\sc acc}'s. If Prover indeed has a DBW for $L_{\neg \infty a}$, then she can follow its transition function and win the game. By following the refuter $\R$, however, Refuter can always fool Prover and generate a word $x$ such that $x \in L_{\neg \infty a}$ iff $y$ contains only finitely many {\sc acc}'s. \hfill $\blacksquare$
\end{example}

We first define refuters for DBW-recognizability, and study their construction and size for languages given by deterministic or nondeterministic automata.
Our refuters serve as a first inexpressibility certificate.
We continue and argue that each DBW-refuter for a language $L$ induces three words $x \in \Sigma^*$ and $x_1,x_2 \in \Sigma^*$, such that $x \cdot (x_1 + x_2)^* \cdot x_1^\omega \subseteq L$ and $x \cdot (x_1^* \cdot x_2)^\omega \cap L = \emptyset$. The triple $\zug{x,x_1,x_2}$ is an additional certificate for $L$ not being in DBW. Indeed, we show that  a language $L$ is not in DBW iff it has a certificate as above. For example,  the language $L_{\neg \infty a}$ has a certificate $\zug{\epsilon,b,a}$. In fact, we show that Landweber's proof for $L_{\neg \infty a}$ can be used as is for all languages not in DBW, with $x_1$ replacing $b$, $x_2$ replacing $a$, and adding $x$ as a prefix.

We then generalize our results on DBW-refutation and certification in two orthogonal directions. The first is an extension to richer classes of deterministic automata, in particular all classes in the two hierarchies discussed above, as well as all deterministic Emerson-Lei automata (DELWs) \cite{EL87}.
For the depth hierarchy, we add to the winning condition of the game a {\em structural restriction}. For example, in a weak automaton, Prover loses if the sequence $y \in A^\omega$ of annotations she generates includes infinitely many alternations between $\acc$ and $\rej$. We show how structural restrictions can be easily expressed in our framework.

The second direction is an extension of the recognizability question to the questions of {\em separation} and {\em approximation}: We say that a language $L \subseteq \Sigma^\omega$ is a {\em separator\/} for two languages $L_1,L_2 \subseteq \Sigma^\omega$ if $L_1 \subseteq L$ and $L \cap L_2 = \emptyset$.
Studies of separation include a search for regular separators of general languages \cite{CLMMKS18}, as well as separation of regular languages by weaker classes of languages, e.g., FO-definable languages \cite{PZ16} or piecewise testable languages \cite{CMM13}. In the context of $\omega$-regular languages, \cite{AS87} presents an algorithm computing the smallest safety language containing a given language $L_1$, thus finding a safety separator for $L_1$ and $L_2$. As far as we know, besides this result there has been no systematic study of separation of $\omega$-regular languages by deterministic automata.

In addition to the interest in separators, we use them in the context of recognizability in two ways. First, a third type of certificate that we suggest for DBW-refutation of a language $L$ are ``simple'' languages $L_1$ and $L_2$ such that
$L_1 \subseteq L$, $L \cap L_2 =\emptyset$, and $\zug{L_1,L_2}$ are not DBW-separable.
Second, we use separability in order to approximate languages that are not in DBW.
Consider such a language $L \subseteq \Sigma^\omega$. A user may be willing to approximate $L$ in order to obtain DBW-recognizability. Specifically, we assume that there are languages $I_{\downarrow} \subseteq L$ and $I_{\uparrow} \subseteq \Sigma^\omega \setminus L$ of words that the user is willing to under- and over-approximate $L$ with. Thus, the user searches for a language that is a separator for $L \setminus I_\downarrow$ and $\Sigma^\omega \setminus (L \cup I_\uparrow)$. We study DBW-separability and DBW-approximation, namely separability and approximation by languages in DBW. In particular, we are interested in finding ``small'' approximating languages $I_\downarrow$ and $I_\uparrow$ with which $L$ has a DBW-approximation, and we show how certificates that refute DBW-separation can direct the search to for successful $I_\downarrow$ and $I_\uparrow$. Essentially, as in {\em counterexample guided abstraction-refinement\/} (CEGAR) for model checking \cite{CGJLV03}, we use certificates for non-DBW-separability in order to suggest interesting {\em radius languages}. While in CEGAR the refined system excludes the counterexample, in our setting the approximation of $L$ excludes the certificate. As has been the case with recognizability, we extend our results to all classes of deterministic automata.

\section{Preliminaries}

\subsection{Transducers and Realizability}

Consider two finite alphabets $\Sigma$ and $A$. It is convenient to think about $\Sigma$ as the \enquote{main} alphabet, and about $A$ as an alphabet of annotations.
For two words $x=x_0 \cdot x_1 \cdot x_2 \cdots \in \Sigma^\omega$ and $y = y_0 \cdot y_1 \cdot y_2 \cdots \in A^\omega$, we define $x \oplus y$ as the word in $(\Sigma \times A)^\omega$ obtained by merging $x$ and $y$. Thus, $x\oplus y = (x_0,y_0)  \cdot (x_1,y_1) \cdot (x_2,y_2) \cdots$.

A {\em $(\Sigma/A)$-transducer\/} models a finite-state system that responds with letters in $A$ while interacting with an environment that generates letters in $\Sigma$. Formally, a $(\Sigma/A)$-transducer is $\T=\zug{\Sigma,A,\init,S,s_0,\rho,\tau}$, where $\init \in \{{\it sys,env}\}$ indicates who initiates the interaction -- the system or the environment, $S$ is a set of states, $s_0 \in S$ is an initial state, $\rho:S \times \Sigma \rightarrow S$ is a transition function, and $\tau:S \rightarrow A$ is a labelling function on the states.
Consider an input word $x=x_0 \cdot x_1 \cdot x_2 \cdots \in \Sigma^\omega$. The {\em run\/} of $\T$ on $x$ is the sequence $s_0,s_1,s_2 \ldots$ such that for all $j \geq 0$, we have that $s_{j+1} = \rho(s_j,x_j)$. The {\em annotation of $x$ by $\T$}, denoted $\T(x)$, depends on $\iota$. If $\init={\it sys}$, then $\T(x) = \tau(s_0) \cdot \tau(s_1) \cdot \tau(s_2) \cdots \in A^\omega$. Note that the first letter in $A$ is the output of $\T$ in $s_0$. This reflects the fact that the system initiates the interaction. If $\init={\it env}$, then $\T(x) = \tau(s_1) \cdot \tau(s_2) \cdot \tau(s_3) \cdots \in A^\omega$. Note that now, the output in $s_0$ is ignored, reflecting the fact that the environment initiates the interaction.

Consider a language $L \subseteq (\Sigma \times A)^\omega$. Let $\comp (L)$ denote the complement of $L$. Thus, $\comp (L) = (\Sigma \times A)^\omega \setminus L$.
We say that a language $L \subseteq (\Sigma \times A)^\omega$ is {\em $(\Sigma/A)$-realizable by the system\/} if there is a $(\Sigma/A)$-transducer $\T$ with $\init={\it sys}$ such that for every word $x \in \Sigma^\omega$, we have that $x \oplus \T(x) \in L$. Then, $L$ is {\em $(A/\Sigma)$-realizable by the environment\/} if there is an $(A/\Sigma)$-transducer $\T$ with $i={\it env}$ such that for every word $y \in A^\omega$, we have that $\T(y) \oplus y \in L$.
When the language $L$ is regular, realizability reduces to deciding a game with a regular winning condition. Then, by determinacy of games and due to the existence of finite-memory winning strategies \cite{BL69}, we have the following.

\begin{proposition}\label{determinacy}
For every $\omega$-regular language $L \subseteq (\Sigma \times A)^\omega$, exactly one of the following holds.
\begin{enumerate}
\item
$L$ is $(\Sigma/A)$-realizable by the system.
\item
$\comp (L)$ is $(A/\Sigma)$-realizable by the environment.
\end{enumerate}
\end{proposition}

\subsection{Automata}
\label{subsec:automata}

A {\em deterministic word automaton} over a finite alphabet $\Sigma$ is $\A=\zug{\Sigma, Q, q_0, \delta, \alpha}$, where $Q$ is a set of states, $q_0 \in Q$ is an initial state, $\delta : Q\times \Sigma \rightarrow Q$ is a transition function, and $\alpha$ is an acceptance condition.
We extend $\delta$ to words in $\Sigma^*$ in the expected way, thus for $q\in Q$, $w \in \Sigma^*$, and letter $\sigma \in \Sigma$, we have that $\delta(q,\epsilon)=q$ and $\delta(q,w \sigma)=\delta(\delta(q,w),\sigma)$.
A {\em run} of $\A$ on an infinite word $\sigma_0,\sigma_1,\dots \in \Sigma^\omega$ is the sequence of states $r = q_0, q_1, \dots$, where for every position $i\geq 0$, we have that $q_{i+1}= \delta(q_i, \sigma_i)$. We use ${\it inf}(r)$ to denote the set of states that $r$ visits infinitely often. Thus, ${\it inf}(r) = \{q: \text{ } q_i=q \text{ for infinitely many }i \geq 0\}$.

The acceptance condition $\alpha$ refers to $\infi(r)$ and determines whether the run $r$ is accepting. For example, in the {\em B\"uchi}, acceptance condition, we have that $\alpha\subseteq Q$, and a run is accepting iff it visits states in $\alpha$ infinitely often; that is, $\alpha\cap {\it inf}(r)\neq\emptyset$. Dually, in {\em co-B\"uchi}, $\alpha\subseteq Q$, and a run is accepting iff it visits states in $\alpha$ only finitely often; that is, $\alpha\cap {\it inf}(r) = \emptyset$. The language of $\A$, denoted $\lang(\A)$, is then the set of words $w$ such that the run of $\A$ on $w$ is accepting.

A parity condition is $\alpha:Q \rightarrow \{0,\ldots,k\}$, for $k \geq 0$, termed the {\em index\/} of $\alpha$. A run $r$ satisfies $\alpha$ iff the maximal color $i \in \{0,\ldots,k\}$ such that $\alpha^{-1}(i)\cap {\it inf}(r)\neq\emptyset$ is odd. That is, $r$ is accepting iff the maximal color that $r$ visits infinitely often is odd. Then, a Rabin condition is $\alpha=\{\zug{G_1,B_1},\ldots,\zug{G_k,B_k}\}$, with $G_i,B_i \subseteq Q$, for all $0 \leq i \leq k$. A run $r$ satisfies $\alpha$ iff there is $1 \leq i \leq k$ such that ${\it inf}(r) \cap G_i \neq\emptyset$ and ${\it inf}(r) \cap B_i = \emptyset$. Thus, there is a pair $\zug{G_i,B_i}$ such that $r$ visits states in $G_i$ infinitely often and visits states in $B_i$ only finitely often.

All the acceptance conditions above can be viewed as special cases of the {\em Emerson-Lei acceptance condition\/} (EL-condition, for short) \cite{EL87}, which we define below. Let $\mathbb{M}$ be a finite set of marks. Given an infinite sequence $\pi = M_0 \cdot M_1 \cdots \in (2^{\mathbb{M}})^\omega$ of subsets of marks, let $\infi(\pi)$ be the set of marks that appear infinitely often in sets in $\pi$. Thus, $\infi(\pi)=\{m \in \mathbb{M}$ : there exist infinitely many $i \geq 0$ such that  $m \in M_i\}$. An EL-condition is a Boolean assertion over atoms in $\mathbb{M}$. For simplicity, we consider assertions in positive normal form, where negation is applied only to atoms. Intuitively, marks that appear positively should repeat infinitely often and marks that appear negatively should repeat only finitely often. Formally,
a deterministic EL-automaton is $\A=\zug{\Sigma, Q, q_0, \delta, \mathbb{M}, \tau, \theta}$, where $\tau \colon Q \to 2^{\mathbb{M}}$ maps each state to a set of marks, and $\theta$ is an EL-condition over $\mathbb{M}$. A run $r$ of a $\A$ is accepting if $\infi(\tau(r))$ satisfies $\theta$.

For example, a B\"uchi condition $\alpha \subseteq Q$ can be viewed as an EL-condition with $\mathbb{M}=\{\acc\}$ and $\tau(q)=\{\acc\}$ for $q \in \alpha$ and $\tau(q)=\emptyset$ for $q \not \in \alpha$. Then, the assertion $\theta=\acc$ is satisfied by sequences $\pi$ induced by runs $r$ with $\infi(r) \cap \alpha \neq \emptyset$. Dually, the assertion $\theta=\neg \rej$ with $\mathbb{M}=\{\rej\}$ is satisfied by sequences $\pi$ induced by runs $r$ with $\infi(r) \cap \alpha = \emptyset$, and thus corresponds to a co-B\"uchi condition. In the case of a parity condition $\alpha : Q \rightarrow \{0,\ldots,k\}$, it is not hard to see that $\alpha$ is equivalent to an EL-condition in which $\mathbb{M} = \{0, 1, \ldots, k\}$,  for every state $q \in Q$, we have that $\tau(q)=\{\alpha(q)\}$, and
\ifarxiv $\theta = \theta_k$ expresses the parity condition, where $\theta_k$ is inductively defined as:
\[\theta_k = \begin{cases}
\neg 0 & \text{if } k = 0, \\
\neg k \wedge \theta_{k - 1} & \text{if } k \text{ is even,} \\
\phantom{\neg} k \vee \theta_{k - 1} & \text{If $k>0$ and $k$ is odd.}
\end{cases}\]
\fi
Lastly, a Rabin condition $\alpha = \{\zug{G_1,B_1},\ldots,\zug{G_k,B_k}\}$ is equivalent to an EL-condition with $\mathbb{M} = \{G_1,B_1, \dots,G_k,B_k\}$ and $\tau(q)=\{m \in \mathbb{M}: q \in m\}$. Note that now, the mapping $\tau$ is not to singletons, and each state is marked by all sets in $\alpha$ in which it is a member. Then, $\theta=\bigvee_{1 \leq i \leq k} (G_i \wedge \neg B_i)$.

We use DBW, DCW, DPW, DRW, DELW to denote deterministic B\"uchi, co-B\"uchi, parity, Rabin, and EL word automata, respectively.
For parity automata, we also use DPW[$0,k$] and DPW[$1,k$], for $k \geq 0$, to denote DPWs in which the colours are in $\{0,\ldots,k\}$ and $\{1,\ldots,k\}$, respectively. For Rabin automata, we use DRW[$k$], for $k \geq 0$, to denote DRWs that have at most $k$ elements in $\alpha$. Finally, we use DELW[$\theta$], to denote DELWs with EL-condition $\theta$.
We sometimes use the above acronyms in order to refer to the set of languages that are recognizable by the corresponding class of automata. For example, we say that a language $L$ is in DBW if $L$ is \emph{DBW-recognizable}, thus there is a DBW $\A$ such that $L=L(\A)$.
Note that DBW = DPW[$0,1$], DCW = DPW[$1,2$], and DRW[$1$] = DPW[$0,2$]. In fact, in terms of expressiveness, DRW[$k$] = DPW[$0,2k$] \cite{Saf92,Lod99b}.

Consider a directed graph $G=\zug{V,E}$. A {\em strongly connected set\/} of $G$ (SCS) is a set $C \subseteq V$ of vertices such that for every two vertices $v,v' \in C$, there is a path from $v$ to $v'$. An SCS $C$ is {\em maximal\/} if it cannot be extended to a larger SCS. Formally, for every nonempty $C' \subseteq V \setminus C$, we have that $C \cup C'$ is not an SCS. The maximal strongly connected sets are also termed {\em strongly connected components\/} (SCC). An automaton $\A=\zug{\Sigma,Q,Q_0,\delta,\alpha}$ induces a directed graph $G_\A=\zug{Q,E}$ in which $\zug{q,q'} \in E$ iff there is a letter $\sigma$ such that $q' \in \delta(q,\sigma)$. When we talk about the SCSs and SCCs of $\A$, we refer to those of $G_\A$. Consider a run $r$ of an automaton $\A$. It is not hard to see that the set ${\it inf}(r)$ is an SCS. Indeed, since every two states $q$ and $q'$ in ${\it inf}(r)$ are visited infinitely often, the state $q'$ must be reachable from $q$.

\ifarxiv
A DBW $\A=\zug{\Sigma, Q, q_0, \delta, \alpha}$ is \emph{weak} (DWW) if every SCC $C$ of $\A$ is accepting, namely $C \subseteq \alpha$, or rejecting, namely $C \cap \alpha =\emptyset$. Thus, each run of $\A$ eventually visits either states in $\alpha$ or only states not in $\alpha$. It is easy to see that every DWW can be viewed as a DBW and as a DCW. In order to refer to the depth of the SCCs in $\A$, we also refer to $\A$ also as a DPW. Indeed, a DPW $\A=\zug{\Sigma, Q, q_0, \delta, \alpha}$ is \emph{weak} if for every transition $q' = \delta(q, \sigma)$ we have $\alpha(q') \geq \alpha(q)$, i.e., $\alpha$ is monotonically increasing along a run. We use DWW[$0,k$] and DWW[$1,k$] to denote weak DPW[$0,k$] and weak DPW[$1,k$], respectively. Finally, note that for each safety $\omega$-regular language $L$, there exists a DWW[$1,2$] that recognises $L$ and all DWW[$1,2$] recognise a safety language. Dually, co-safety languages correspond to DWW[$0,1$].
\fi

\section{Refuting DBW-Recognizability}

Let $A=\{\acc,\rej\}$. We use $\infty \acc$ to denote the subset $\{a_0 \cdot a_1 \cdot a_2 \cdots \in A^\omega : \mbox{ there are infinitely many $j \geq 0$ with } a_j = \acc\}$ and $\neg \infty \acc=\comp (\infty \acc) = \{a_0 \cdot a_1 \cdot a_2 \cdots \in A^\omega : \mbox{ there are only finitely many $j \geq 0$ with } a_j = \acc\}$.

A DBW $\A=\zug{\Sigma, Q, q_0, \delta, \alpha}$ can be viewed as a $(\Sigma/A)$-transducer $\T_\A=\langle \Sigma,A,{\it sys}$, $Q$, $q_0,\delta,\tau\rangle$, where for every state $q \in Q$, we have that $\tau(q)=\acc$ if $q \in \alpha$, and $\tau(q)=\rej$ otherwise.
Then, for every word $x \in \Sigma^\omega$, we have that $x \in L(\A)$ iff $\T_\A(x) \in \infty{\acc}$.

For a language $L \subseteq \Sigma^\omega$, we define the language $\DBWreal(L) \subseteq (\Sigma \times A)^\omega$ of words with correct annotations. Thus,
\[\DBWreal(L)=\{x \oplus y : x \in L \mbox{ iff } y \in \infty{\acc}\}.\]
Note that $\comp (\DBWreal(L))$ is the language
\[\NoDBWreal(L)=\{x \oplus y : (x \in L \mbox{ and } y \not \in \infty{\acc}) \mbox{ or }
 (x \not \in L \mbox{ and }  y  \in \infty{\acc})\}.\]
A {\em DBW-refuter for $L$\/} is an $(A/\Sigma)$-transducer with $\init={\it env}$ realizing $\NoDBWreal(L)$.

\begin{example}
\label{rom refuter}
For every language $R \subseteq \Sigma^*$ of finite words, the language $R^\omega \subseteq \Sigma^\omega$ consists of infinite concatenations of words in $R$. It was recently shown that $R^\omega$ may not be in DBW \cite{LK20}. The language used in \cite{LK20} is $R=\$+(0\cdot \{0,1,\$\}^*\cdot 1)$. In \Cref{ref romega} below we describe a DBW-refuter for $R^\omega$.
\begin{figure}[h]
\begin{center}
\vspace{-1em}
\includegraphics[scale=1]{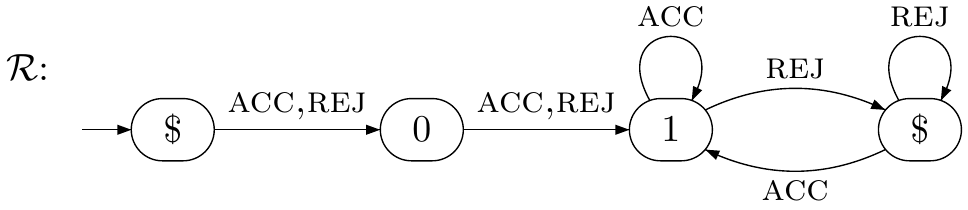}
\caption{A DBW-refuter for $(\$+(0\cdot \{0,1,\$\}^*\cdot 1))^\omega$.}
\vspace{-2em}
\label{ref romega}
\end{center}
\end{figure}

Following $\R$, Refuter starts by generating a prefix $0 \cdot 1$ and then responds to $\acc$ with $1$ and responds with $\$$ to $\rej$. Accordingly, if Prover generates a rejecting run, Prover generates a word in $0 \cdot 1 \cdot (1 + \$)^* \cdot \$^\omega$, which is in $R^\omega$. Also, if Prover generates an accepting run, Prover generates a word in $0 \cdot 1 \cdot (1^+ \cdot \$^*)^\omega$, which has a single $0$ and infinitely many $1$'s, and is therefore not in $R^\omega$.
\hfill $\blacksquare$ \end{example}

By \Cref{determinacy}, we have the following.
\begin{proposition}\label{dbw:recognisability}
Consider a language $L \subseteq \Sigma^\omega$. Let $A=\{\acc,\rej\}$. Exactly one of the following holds:
\begin{itemize}
\item
$L$ is in DBW, in which case the language $\DBWreal(L)$ is $(\Sigma/A)$-realizable by the system, and a finite-memory winning strategy for the system induces a DBW for $L$.
\item
$L$ is not in DBW, in which case the language $\NoDBWreal(L)$ is $(A/\Sigma)$-realizable by the environment, and a finite-memory winning strategy for the environment induces a DBW-refuter for $L$.
\end{itemize}
\end{proposition}

\subsection{Complexity}
In this section we analyze the size of refuters. We start with the case where the language $L$ is given by a DPW.

\begin{theorem}
\label{dbw by dpw yes or refute}
Consider a DPW $\A$ with $n$ states. Let $L=L(\A)$. One of the following holds.
\begin{enumerate}
\item
There is a DBW for $L$ with $n$ states.
 \item
There is a DBW-refuter for $L$ with $2n$ states.
 \end{enumerate}
\end{theorem}

\begin{proof}
If $L$ is in DBW, then, as DPWs are B\"uchi type \cite{KPB94}, a DBW for $L$ can be defined on top of the structure of $\A$, and so it has $n$ states.
If $L$ is not in DBW, then by \Cref{dbw:recognisability}, there is a DBW-refuter for $L$, namely  a $(\{\acc,\rej\}/\Sigma)$-transducer that realizes $\NoDBWreal(L)$. We show we can define a DRW $\U$ with $2n$ states for
$\NoDBWreal(L)$. The result then follows from the fact a realizable DRW is realized by a transducer of the same size as the DRW \cite{EJ88}.

We construct $\U$ by taking the union of the acceptance conditions of a DRW $\U_1$ for $\{x \oplus y : x \in L \mbox{ and } y \not \in \infty{\acc}\}$ and a DRW $\U_2$ for $\{x \oplus y : x \not \in L \mbox{ and }  y  \in \infty{\acc}\}$. We obtain both DRWs by taking the product of $\A$, extended to the alphabet $\Sigma \times \{\acc,\rej\}$, with a $2$-state automaton for $\infty{\acc}$, again extended to the alphabet $\Sigma \times \{\acc,\rej\}$.

We describe the construction in detail. Let $\A=\zug{\Sigma,Q,q_0,\delta,\alpha}$. Then, the state space of $\U_1$ is $Q \times \{\acc,\rej\}$ and its transition on a letter $\zug{\sigma,a}$ follows $\delta$ when it reads $\sigma$, with $a$ determining whether $\U_1$ moves to the $\acc$ or $\rej$ copy. Let $\alpha_1$ be the Rabin condition equivalent to $\alpha$. We obtain the acceptance condition of $\U_1$ by replacing each pair $\zug{G,B}$ in $\alpha_1$ by $\zug{G \times \{\rej\},B \times \{\rej\} \cup Q \times \{\acc\}}$. It is not hard to see that a run of $\U_1$ satisfies the latter pair iff its projection on $Q$ satisfies the pair $\zug{G,B}$ and its projection on $\{\acc,\rej\}$ has only finitely many $\acc$. The construction of $\U_2$ is similar, with $\alpha_2$ being a Rabin condition that complements $\alpha$, and then replacing each pair $\zug{G,B}$ in $\alpha_2$ by $\zug{G \times \{\acc\}, B \times \{\acc,\rej\})}$. Since $\U_1$ and $\U_2$ have the same state space, and we only have to take the union of the pairs in their acceptance conditions, the $2n$ bound follows.
\hfill \qed
\end{proof}

Now, when $L$ is given by an NBW, an exponential bound follows from the exponential blow up in determinization \cite{Saf88}. If we are also given an NBW for $\comp(L)$, the complexity can be tightened. Formally, we have the following.

\begin{theorem}
\label{dbw yes or refute}
Given NBWs with $n$ and $m$ states, for $L$ and $\comp (L)$, respectively, one of the following holds.
\begin{enumerate}
\item
There is a DBW for $L$ with $\min\{(1.65n)^n, 3^m\}$ states.
 \item
There is a DBW-refuter for $L$ with $\min\{2 \cdot (1.65n)^n, 2 \cdot (1.65m)^m\}$ states.
 \end{enumerate}
\end{theorem}

\begin{proof}
If $L$ is in DBW, then a DBW for $L$ can be defined on top of a DPW for $L$, which has at most $(1.65n)^n$ states \cite{Sch09}, or by dualizing a DCW for $\comp (L)$. Since the translation of an NBW with $m$ states to a DCW, when it exists, results in a DCW with $3^m$ states \cite{BK09}, we are done.
If $L$ is not in DBW, then we proceed as in the proof of \Cref{dbw by dpw yes or refute}, defining $\U$ on the top of a DPW for either $L$ or $\comp(L)$.
\hfill \qed \end{proof}

\subsection{Certifying DBW-Refutation}

Consider a DBW-refuter $\R=\zug{\{\acc,\rej\},\Sigma,{\it env},S,s_0,\rho,\tau}$. We say that a path $s_0,\ldots,s_m$ in $\R$ is an $\rej^+$-path if it contains at least one transition and all the transitions along it are labeled by $\rej$; thus, for all $0 \leq j <m$, we have that $s_{j+1}=\rho(s_j,\rej)$. Then, a path $s_0,\ldots,s_m$ in $\R$ is an $\acc$-path if it contains at least one transition and its first transition is labeled by $\acc$. Thus, $s_1 = \rho(s_0,\acc)$.

\begin{lemma}\label{lem: lan struct}
Consider a DBW-refuter $\R=\zug{\{\acc,\rej\},\Sigma,{\it env},S,s_0,\rho,\tau}$. Then there exists a state $s \in S$, a (possibly empty) path $p = s_0,s_1, \dots s_m$, a $\rej^+$-cycle $p_1 = s^1_0,s^1_{1} \dots s^{1}_{m_1}$, and an $\acc$-cycle  $p_2 = s^2_0,s^2_{1} \dots s^{2}_{m_2}$, such that $s_m=s^1_0=s^{1}_{m_1}=s^2_0=s^{2}_{m_2}=s$.
\end{lemma}

\begin{proof}
Let $s_{i} \in S$ be a reachable state that belongs to an ergodic component in the graph of $\R$ (that is, $s_{i} \in C$, for a set $C$ of strongly connected states that can reach only states in $C$).
Since $\R$ is responsive, in the sense it can read in each round both $\acc$ and $\rej$, we can read from $s_i$ the input sequence $\rej^\omega$. Hence, $\R$ has a $\rej^+$-path $s_i,\ldots,s_l,\ldots,s_k$ with $s_l=s_k$, for $l < k$. It is easy to see that the claim holds with $s=s_l$. In particular, since $\R$ is responsive and $C$ is strongly connected, there exists an $\acc$-cycle from $s_l$ to itself.
\hfill \qed \end{proof}

\ifarxiv
\begin{figure}[hbt]
  \begin{center}
  	\includegraphics{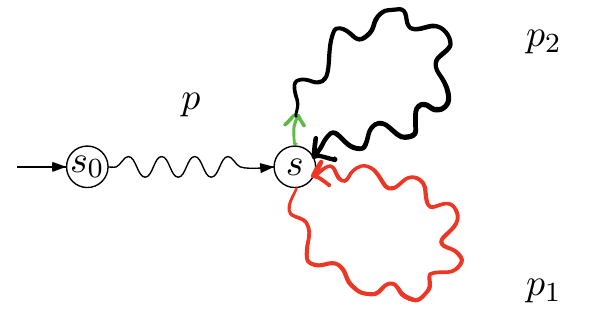}
  	\vspace{-3mm}
  \end{center}
  \caption{The structure from \Cref{lem: lan struct} that exists in every DBW-refuter.}
\end{figure}
\fi

\begin{theorem}\label{lem:not-dbw-cert}
An $\omega$-regular language $L$ is not in DBW iff there exist three finite words $x \in \Sigma^*$ and $x_1, x_2 \in \Sigma^+$, such that
\ifarxiv
\[x \cdot (x_1 + x_2)^* \cdot x_1^\omega \subseteq L \quad \text{ and } \quad x \cdot (x_1^* \cdot x_2)^\omega \cap L = \emptyset.\]
\fi
\end{theorem}

\begin{proof}
Assume first that $L$ is not in DBW. Then, by \Cref{dbw yes or refute}, there exists a DBW-refuter $\R$ for it. Let $p = s_0,s_1, \dots s_m$, $p_1 = s^1_0,s^1_{1}, \dots, s^{1}_{m_1}$, and $p_2 = s^2_0,s^2_{1}, \dots, s^{2}_{m_2}$, be the path, $\rej^+$-cycle, and  $\acc$-cycle that are guaranteed to exist by \Cref{lem: lan struct}. Let $x,x_1$, and $x_2$ be the outputs that $\R$ generates along them. Formally, $x=\tau(s_1)\cdot \tau(s_2)\cdots\tau(s_{m})$, $x_1=\tau(s^1_{1}) \cdot \tau(s^1_{2})\cdots \tau(s^{1}_{m_1})$, and
$x_2=\tau(s^2_1) \cdot \tau(s^2_{1})\cdots \tau(s^{2}_{m_2})$. Note that as the environment initiates the interaction, the first letter in the words $x$, $x_1$, and $x_2$, are the outputs in the second states in $p$, $p_1$, and $p_2$.
\ifarxiv
We prove that $x,x_1$, and $x_2$ satisfy the two conditions in the theorem.

Let $y \in \{\acc,\rej\}^*$, and $y_1,y_2\in \{\acc,\rej\}^+$ be the input sequences read along $p,p_1$, and $p_2$, respectively. Thus, $y=a_0,a_1,\ldots,a_{m-1}$ is such that for all $0 \leq j < m$, we have that $s_{j+1}=\rho(s_j,a_j)$, and similarly for $y_1$ and $y_2$ with $p_1$ and $p_2$.

Consider a word $w \in x\cdot (x_1 + x_2)^* \cdot x_1^\omega$. Let $a \in y \cdot (y_1 + y_2)^* \cdot y_1^\omega$ be such that $\R(a)=w$. Note we can obtain $a$ from $w$ by replacing each subword $x$ by $y$, $x_1$ by $y_1$, and $x_2$ by $y_2$. Since $p_1$ is a $\rej^+$-cycle, we have that $a \in (\acc+\rej)^* \cdot \rej^\omega$, and so $a \in \neg \infty \acc$. Since $\R$ is a refuter for $L$, it follows that $\R(a) \in L$. Hence, $x \cdot (x_1 + x_2)^* \cdot x_1^\omega \subseteq L$.

For this direction it remains to show that $x \cdot (x_1^* \cdot x_2)^\omega \cap L = \emptyset$. Consider a word $w \in x \cdot (x_1^* \cdot x_2)^\omega$, and let $a \in y \cdot (y_1^* \cdot y_2)^\omega$ be such that $\R(a)=w$. Since $p_1$ is an $\acc$-cycle, we have that $a \in (\rej^* \acc)^\omega$, and so $a \in \infty \acc$. Since $\R$ is a refuter for $L$, it follows that $\R(a) \notin L$. Hence, $x \cdot (x_1^* \cdot x_2)^\omega \cap L = \emptyset$, and we are done.
\fi


For the other direction, we adjust Landweber's proof \cite{Lan69} for the non-DBW-recogniz\-ability of $\neg \infty a$ to $L$. Essentially, $\neg \infty a$ can be viewed as a special case of $x \cdot (x_1 + x_2)^* \cdot x_1^\omega$, with $x=\epsilon$, $x_1=b$, and $x_2=a$. Assume by way of contradiction that there is a DBW $\A$ with $\lang(\A)=L$.
Let $\A=\zug{\Sigma,Q,q_0,\delta,\alpha}$.
Consider the infinite word $w_0 = x \cdot x_1^\omega$. Since $w_0 \in x \cdot (x_1 + x_2)^* \cdot x_1^\omega$, and so $w \in L$, the run of $\A$ on $w_0$ is accepting.
Thus, there is $i_1 \geq 0$ such that $\A$ visits $\alpha$ when it reads the $x_1$ suffix of $x \cdot x_1^{i_1}$.
Consider now the infinite word $w_1 = x \cdot x_1^{i_1} \cdot x_2 \cdot x_1^\omega$.
Since $w_1$ is also in $L$, the run of $\A$ on $w_1$ is accepting.
Thus, there is $i_2 \geq 0$ such that $\A$ visits $\alpha$ when it reads the $x_1$ suffix of $x \cdot x_1^{i_1} \cdot x_2 \cdot x_1^{i_2}$.
In a similar fashion we can continue to find indices $i_1,i_2,\ldots$ such for all $j \geq 1$, we have that
$\A$ visits $\alpha$ when it reads the $x_1$ suffix of $x \cdot x_1^{i_1} \cdot x_2 \cdot x_1^{i_2} \cdot x_2 \cdots x_2 \cdot x_1^{i_j}$.
\ifarxiv
Since $Q$ is finite, there are iterations $j$ and $k$, such that $1 \leq j < k \leq
|Q|+1$ and there is a state $q$ such that $q=\delta(q_0,x \cdot x_1^{i_1} \cdot x_2 \cdot x_1^{i_2} \cdot x_2 \cdots x_2 \cdot x_1^{i_j})=
\delta(q_0,x \cdot x_1^{i_1} \cdot x_2 \cdot x_1^{i_2} \cdot x_2 \cdots x_2 \cdot x_1^{i_k})$.
Since $j < k$, the extension $x_2 \cdot x_1^{i_{j+1}}\cdots x_1^{i_{k-1}}\cdot x_2 \cdot x_1^{i_k}$ is not empty and at least one state in $\alpha$ is visited when $\A$ loops in $q$ while reading it.
It follows that the run of $\A$ on the word
\[w=x \cdot x_1^{i_1} \cdot x_2 \cdot x_1^{i_2} \cdot x_2 \cdots x_2 \cdot
x_1^{i_j} \cdot (x_2 \cdot x_1^{i_{j+1}}\cdots x_1^{i_{k-1}} \cdot x_2 \cdot x_1^{i_k})^\omega\]
is accepting. But $w \in x \cdot (x_1^* \cdot x_2)^\omega$, so it is
not in $L$, and we have reached a contradiction.
\fi
\hfill \qed \end{proof}

\ifarxiv
\begin{remark}
\Cref{lem:not-dbw-cert}, as well as the yet to be presented \Cref{lem:not-dpw-0-cert,thm:not-dww-i-j-cert} are special cases of \cite[Lemma 14]{Wag79}. However, our alternative proof relies on \Cref{determinacy} and the analysis of the resulting refuter, while the proof of \cite{Wag79} examines the structure of a deterministic Muller automaton. Due to the game-based setting we can easily extend our approach to refuting separability of languages (\Cref{section:separability}), which requires substantial modifications of the approach from \cite{Wag79}.
\end{remark}
\fi
We refer to a triple $\zug{x,x_1,x_2}$ of words that satisfy the conditions in \Cref{lem:not-dbw-cert} as a {\em certificate\/} to the non-DBW-recognizability of $L$.

\begin{example}
\label{ex certificate}
In \Cref{rom refuter}, we described a DBW-refuter for $L=(\$+(0\cdot \{0,1,\$\}^*\cdot 1))^\omega$. A certificate to its non-DBW-recognizability is $\zug{x,x_1,x_2}$, with $x=01$, $x_1=\$$, and $x_2=1$. Indeed, $01 \cdot (\$ + 1)^* \cdot \$^\omega \subseteq L$ and $01 \cdot (\$^* \cdot 1)^\omega \cap L = \emptyset$. \hfill $\blacksquare$
\end{example}

Note that obtaining certificates according to the proof of \Cref{lem:not-dbw-cert} may not give us the shortest certificate. For example, for $L$ in \Cref{ex certificate}, the proof would give us $x=01\$$, $x_1=\$$, and $x_2=1\$,$ with $01\$ \cdot (\$ + 1\$)^* \cdot \$^\omega \subseteq L$ and $01\$ \cdot (\$^* \cdot 1\$)^\omega \cap L = \emptyset$. The problem of generating smallest certificates is related to the problem of finding smallest witnesses to DBW non-emptiness \cite{KS06} and is harder. Formally, defining the length of a certificate $\zug{x,x_1,x_2}$ as $|x|+|x_1|+|x_2|$, we have the following:

\begin{theorem}\label{shortest}
Consider a DPW $\A$ and a threshold $l \geq 1$. The problem of deciding whether there is a certificate of length at most $l$ for non-DBW-recognizability of $L(\A)$ is NP-complete, for $l$ given in unary or binary.
\end{theorem}
\ifarxiv
\begin{proof}
We start with membership in NP. Let $n$ be the number of states in $\A$. By \Cref{dbw by dpw yes or refute} and the construction in \Cref{lem:not-dbw-cert} we can bound the length of a certificate to be at most $6n$, since these are constructed from simple paths. Given a witness certificate $\zug{x,x_1,x_2}$ of length at most $l$ (the latter can be checked in polynomial time, regardless of how $l$ is given), checking the conditions in \Cref{lem:not-dbw-cert} involves checking $x \cdot (x_1 + x_2)^* \cdot x_1^\omega \subseteq L(\A)$, namely containment of a DCW of size linear in the certificate in the language of a DPW, which can be done in polynomial time, and checking $x \cdot (x_1^* \cdot x_2)^\omega \cap L(\A) = \emptyset$, namely emptiness of the intersection with a DBW, which again can be done in polynomial time.

For the NP-hardness, we describe a reduction from the Hamiltonian-cycle problem on directed graphs.
Formally, given a directed graph $G = \zug{V,E}$, we describe a DPW that is not in DBW and which has a certificate of length $|V|+1$ iff $G$ has a Hamiltonian cycle, namely a cycle that visits each vertex in $V$ exactly once. The proof elaborates on the NP-hardness proof of the problem of finding a shortest witness to DBW non-emptiness \cite{KS06}.

Let $V = \{1, \ldots, n\}$, and assume that $n \geq 2$ and $E$ is not empty.
We define a DPW $\A = \zug{E, (V \times V) \cup \{\zug{1,1}_{\text{err}}\}, \{\zug{1,1}\}, \delta, \alpha}$, where $\alpha(\zug{n,n}) = 1$, $\alpha(\zug{1,1}_{\text{err}}) = 2$, $\alpha(q) = 0$ for all other states $q$, and
\begin{align*}
\delta(\zug{i,j},(k,h)) & = \begin{cases}
	\zug{h, (j\text{ mod }n) + 1} & \text{if } i = k = j, \\
	\zug{h,j}                     & \text{if } i = k \neq j, \\
	\zug{1,1}_{\text{err}}		  & \text{otherwise.}
\end{cases} \\
\delta(\zug{1,1}_{err},(k,h)) & = \begin{cases}
	\zug{h,2} \hspace{5.75em}        & \text{if } k = 1, \\
	\zug{1,1}_{\text{err}}		  & \text{otherwise.}
\end{cases}
\end{align*}

Intuitively, $\A$ interprets a word $w \in E^\omega$, as an infinite path starting in vertex $1$, and it verifies that the path is valid on $G$. Whenever $\A$ encounters an edge that does not match the current state, which is tracked in the first component of the state space, it resets and moves to $\zug{1,1}_{\text{err}}$. The second component of a state $\zug{i,j}$ is the vertex the path owes a visit in order to visit all vertices infinitely often. It is easy to see that $w \in L(\A)$ iff there is a suffix $w'$ of $w$ that describes a valid path in $G$ that visits every vertex infinitely often. Notice that $L(\A)$ is not DBW-recognizable and that $\A$ is polynomial in the size of $G$.

Clearly, the reduction is polynomial, we now prove its correctness. Assume first that $G$ has a Hamiltonian cycle $c$. Then, from the word $w$ read along $c$ from vertex 1, we construct the certificate $\zug{\epsilon, w, (2,1)}$ showing non-DBW-recognizabilty. Indeed, the certificate is correct, since $(w + (2,1))^* \cdot w^\omega \subseteq L(\A)$ and $(w^* \cdot (2,1))^\omega \cap L(\A) = \emptyset$. This certificate has size $n+1$.

For the other direction, assume that $\zug{x, x_1, x_2}$ is a certificate of size (at most) $n+1$. Then, $x \cdot x_1^\omega \in L(\A)$ and as $x_2$ is not empty, it must be that $|x| + |x_1| \leq n$. Let $r$ be the corresponding accepting run and thus $r$ visits $\zug{n,n}$ infinitely often. By the definition of $\delta$, the run $r$ also visits the states $\zug{i,i}$, for all $1 \leq i \leq n$. Since the transitions to each of these states are labelled differently, $x_1$ must contain at least $n$ different letters. Hence, $|x_1|$ must be $n$ and thus $G$ has a Hamiltonian cycle.
\end{proof}
\fi

\begin{remark}
\label{remark sccs}{\bf [Relation with existing characterizations]}\
By \cite{Lan69}, the language of a DPW $\A=\zug{\Sigma,Q,q_0,\delta,\alpha}$ is in DBW iff for every accepting SCS $C \subseteq Q$ and SCS $C' \supseteq C$, we have that $C'$ is accepting. The proof of Landweber relies on a complicated analysis of the structural properties of $\A$. As we elaborate \ifarxiv below, \fi \Cref{lem:not-dbw-cert}, which relies instead on determinacy of games, suggests an alternative proof. Similarly, \cite{Wag79} examines the structure of a deterministic Muller automaton, and
\Cref{lem:not-dbw-cert} can be viewed as a special case of Lemma 14 there, with a proof based on the game setting. %

\ifarxiv
We use certificates in order to prove that a DPW $\A=\zug{\Sigma,Q,q_0,\delta,\alpha}$ is in DBW iff for every accepting SCS $C \subseteq Q$ and SCS $C' \supseteq C$, we have that $C'$ is accepting.
First, an accepting SCS $C \subseteq Q$ and a rejecting SCS $C' \supseteq C$ induce a certificate $\zug{x,x_1,x_2}$. Indeed, taking a state $s \in C$, we can define $x$ to be a word that leads from $q_0$ to $s$, $x_1$ to be a word that traverses $C$, and $x_2$ a word that traverses $C'$.
Then, the set of states traversed infinitely often in a run on a word in $x \cdot (x_1 + x_2)^* \cdot x_1^\omega$ is $C$, and the set of states traversed infinitely often in a run on a word in $x \cdot (x_1^* \cdot x_2)^\omega$ is $C'$. For the other direction, a certificate $\zug{x,x_1,x_2}$ induces an accepting SCS $C \subseteq Q$ and a rejecting SCS $C' \supseteq C$ as follows.
Consider a graph $G=\zug{Q,E}$, where $E(s,s')$ iff $\delta(s,x_1)=s'$ or $\delta(s,x_2)=s'$. We consider an ergodic SCC that is reachable from $\delta(q_0,x)$ in $G$. In this ergodic SCC, we can traverse both words in $x \cdot (x_1 + x_2)^* \cdot x_1^\omega$ along an accepting cycle $C$, and words in $x \cdot (x_1^* \cdot x_2)^\omega$ along a rejecting cycle, whose union with $C$ can serve as $C'$.
\fi
\hfill $\blacksquare$
\end{remark}

Being an $(A/\Sigma)$-transducer, every DBW-refuter $\R$ is responsive and may generate many different words in $\Sigma^\omega$. Below we show that we can leave $\R$ responsive and yet let it generate only words induced by a certificate. Formally, we have the following.

\begin{lemma}
\label{ref from cert}
Given a certificate $\zug{x,x_1,x_2}$ to non-DBW-recognizability of a language $L \subseteq \Sigma^\omega$, we can define a refuter $\R$ for $L$ such that for every $y \in A^\omega$, if $y \models \infty \acc$, then $\R(y) \in x \cdot (x_1^* \cdot x_2)^\omega$, and if $y \models \neg \infty \acc$, then $\R(y) \in
x \cdot (x_1 + x_2)^* \cdot x_1^\omega$.
\end{lemma}

\begin{proof}
Intuitively, $\R$ first ignores the inputs and outputs $x$. It then repeatedly outputs either $x_1$ or $x_2$, according to the following policy: in the first iteration, $\R$ outputs $x_1$. If during the output of $x_1$ all inputs are $\rej$, then $\R$ outputs $x_1$ also in the next iteration. If an input $\acc$ has been detected, thus the prover tries to accept the constructed word, the refuter outputs $x_2$ in the next iteration, again keeping track of an $\acc$ input. If no $\acc$ has been input, $\R$ switches back to outputting $x_1$.
\ifarxiv

Formally, let $\zug{x, x_1, x_2}$ be a certificate with $x = x^1 \cdots x^n$, $x_1 = x^1_1 \cdots x^{n_1}_1$, and $x_2 = x^1_2 \cdots x_2^{n_2}$. We define $\R= \zug{\{\acc, \rej\}, \Sigma, {\it env}, S, s_0, \rho, \tau}$ with the components $S$, $\rho$, and $\tau$ defined as follows:
\begin{itemize}
	\item $S = \{s_0, s_1, \dots, s_n, (s^1_1,a), \dots, (s^{n_1}_1,a), (s^1_2,a), \dots, (s^{n_2}_2,a) : a \in \{\acc, \rej\}\}$
	\item $\rho(s, a) = \begin{cases}
	s_1               & \text{if } s = s_0 \text{ and } n > 0, \\
	s_{i+1}           & \text{if } s = s_i \text{ and } n > i > 0, \\
	(s^1_1, \rej)     & \text{if } s = s_n, \\
	(s^1_1, \rej)     & \text{if } s \in \{(s^{n_1}_1, \rej), (s^{n_2}_2, \rej)\}\text{ and } a = \rej, \\
	(s^1_2, \rej)     & \text{if } s \in \{(s^{n_1}_1, \rej), (s^{n_2}_2, \rej)\}\text{ and } a = \acc, \\
	(s^1_2, \rej)     & \text{if } s \in \{(s^{n_1}_1, \acc), (s^{n_2}_2, \acc)\} \\
	(s^{i+1}_1, \rej) & \text{if } s = (s^i_1, \rej) \text{ and } n_1 > i > 0 \text{ and } a = \rej, \\
	(s^{i+1}_1, \acc) & \text{if } s = (s^i_1, \rej) \text{ and } n_1 > i > 0 \text{ and } a = \acc, \\
	(s^{i+1}_1, \acc) & \text{if } s = (s^i_1, \acc) \text{ and } n_1 > i > 0, \\
	(s^{i+1}_2, \rej) & \text{if } s = (s^i_2, \rej) \text{ and } n_2 > i > 0 \text{ and } a = \rej, \\
	(s^{i+1}_2, \acc) & \text{if } s = (s^i_2, \rej) \text{ and } n_2 > i > 0 \text{ and } a = \acc, \\
	(s^{i+1}_2, \acc) & \text{if } s = (s^i_2, \acc) \text{ and } n_2 > i > 0.
	\end{cases}$
	\item $\tau(s_i) = x^i$ and $\tau((s^i_j,a)) = x^i_j$. \hfill \qed
\end{itemize}
\fi
\end{proof}

By \Cref{lem:not-dbw-cert}, every language not in DBW has a certificate $\zug{x,x_1,x_2}$.
As we argue below, these certificates are linear in the number of states of the refuters.

\begin{lemma}
\label{cert given refuter}
Let $\R$ be a DBW-refuter for $L \subseteq \Sigma^\omega$ with $n$ states. Then, $L$ has a certificate of the form $\zug{x,x_1,x_2}$ such that $|x|+|x_1|+|x_2| \leq 2\cdot n$.
\end{lemma}

\begin{proof}
The paths $p$, $p_1$, and $p_2$ that induce $x$, $x_1$ and $x_2$ in the proof of \Cref{lem:not-dbw-cert} are simple, and so they are all of length at most $n$. Also, while these paths may share edges, we can define them so that each edge appears in at most two paths. Indeed, if an edge appears in all three path, we can shorten $p$. Hence,
$|x|+|x_1|+|x_2| \leq 2\cdot n$, and we are done.
\hfill \qed \end{proof}

\begin{theorem}
Consider a language $L \subseteq \Sigma^\omega$ not in DBW. The length of a certificate for the non-DBW-recognizability of $L$ is linear in a DPW for $L$ and is exponential in an NBW for $L$. These bounds are tight.
\end{theorem}

\begin{proof}
The upper bounds follow from \Cref{dbw by dpw yes or refute} and \Cref{cert given refuter}, and the exponential determinization of NBWs. The lower bound in the NBW case follows from the exponential lower bound on the size of shortest non-universality witnesses for non-deterministic finite word automata (NFW) \cite{MS72}. We sketch the reduction: Let $L_n \subseteq \{0,1\}^*$ be a language such that the shortest witness for non-universality of $L_n$ is exponential in $n$, but $L_n$ has a polynomial sized NFW. We then define $L'_n = (L_n \cdot \$ \cdot (0^* \cdot 1)^\omega) + ((0+1)^* \cdot \$ \cdot (0+1)^* \cdot 0^\omega)$. It is clear that $L'_n$ has a NBW polynomial in $n$ and is not DBW-recognizable. Note that for every word $w \in L_n$, we have $w \cdot \$ \cdot (0+1)^\omega \subseteq L_n'$. Thus, in order to satisfy \Cref{lem:not-dbw-cert}, every certificate $\zug{x, x_1, x_2}$ needs to have $w \cdot \$ $ as prefix of $x$, for some $w \notin L_n$. Hence, it is exponential in the size of the NBW.
\hfill \qed
\end{proof}

\begin{remark}{\bf [LTL]}\
When the language $L$ is given by an LTL formula $\varphi$, then $\DBWreal(\varphi) = \varphi \leftrightarrow \textbf{GF} \acc$ and thus an off-the-shelf LTL synthesis tool can be used to extract a DBW-refuter, if one exists. As for complexity, a doubly-exponential upper bound on the size of a DPW for $\NoDBWreal(L)$, and then also on the size of DBW-refuters and certificates, follows from the double-exponential translation of LTL formulas to DPWs \cite{VW94,Saf88}. The length of certificates, however, and then, by \Cref{ref from cert}, also the size of a minimal refuter, is related to the {\em diameter\/} of the DPW for $\NoDBWreal(L)$, and we leave its tight bound open.\hfill $\blacksquare$
\end{remark}

\section{Separability and Approximations}
\label{section:separability}

Consider three languages $L_1, L_2, L \subseteq \Sigma^\omega$. We say that $L$ is a {\em separator\/} for $\zug{L_1,L_2}$ if $L_1 \subseteq L$ and $L_2 \cap L = \emptyset$.
We say that a pair of languages $\zug{L_1,L_2}$ is {\em DBW-separable\/} iff there exists a language  $L$ in DBW such that  $L$ is a separator for $\zug{L_1,L_2}$.

\begin{example}
Let $\Sigma = \{a,b\}$, $L_1 = (a + b)^* \cdot b^\omega$, and $L_2 = (a + b)^* \cdot a^\omega$. By \cite{Lan69}, $L_1$ and $L_2$ are not in DBW. They are, however, DBW-separable. A witness for this is $L = (a^* \cdot b)^\omega$. Indeed, $L_1 \subseteq L$, $L \cap L_2 = \emptyset$, and $L$ is DBW-recognizable.
\hfill $\blacksquare$ \end{example}

Consider a language $L \subseteq \Sigma^\omega$, and suppose we know that $L$ is not in DBW. A user may be willing to approximate $L$ in order to obtain DBW-recognizability. Specifically, we assume that there is a language $I \subseteq \Sigma^\omega$ of words that the user is {\em indifferent\/} about. Formally, the user is satisfied with a language in DBW that agrees with $L$ on all words that are not in $I$. Formally, we say that a language $L'$ {\em approximates $L$ with radius $I$} if $L \setminus I \subseteq L' \subseteq L \cup I$. It is easy to see that, equivalently, $L'$ is a separator for $\zug{L \setminus I, \comp(L \cup I)}$. Note that the above formulation embodies the case where the user has in mind different over- and under-approximation radiuses, thus separating $\zug{L \setminus I_{\downarrow}, \comp(L \cup I_{\uparrow})}$ for possibly different $I_{\downarrow}$ and $I_{\uparrow}$. Indeed, by defining $I=(I_\downarrow \cap L) \cup (I_\uparrow \setminus L)$, we get $\zug{L \setminus I, \comp(L \cup I)}=\zug{L \setminus I_{\downarrow}, \comp(L) \setminus I_{\uparrow})}$.

It follows that by studying DBW-separability, we also study DBW-approx\-imation, namely approximation by a language that is in DBW, possibly with different over- and under-approximation radiuses.

\ifarxiv
\begin{figure}[bt]
  \begin{center}
  	\includegraphics{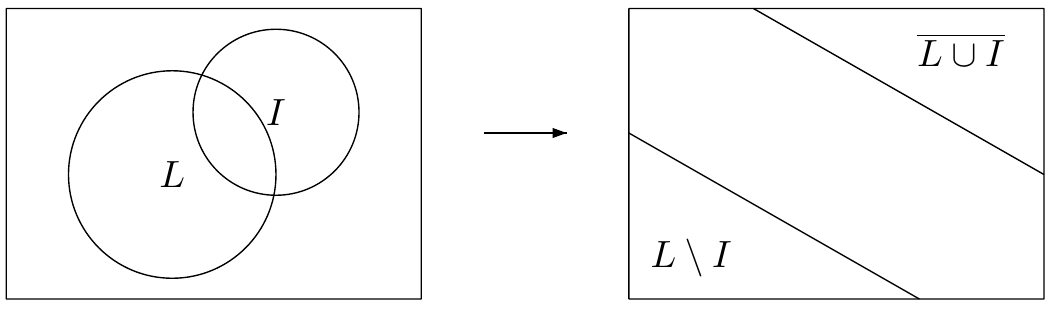}
  \end{center}
  \caption{Reduction of approximation to separability.}
  \label{sep fig}
\end{figure}
\fi

\begin{remark}{\bf [From recognizability to separation]}\
\label{core lang}
It is easy to see that DBW-sep\-arability generalizes DBW-recognizability, as $L$ is in DBW iff $\zug{L, \comp (L)}$ is DBW-separable.
Given $L \subseteq \Sigma^\omega$, we say that a pair of languages $\zug{L_1,L_2}$ is a {\em no-DBW-witness\/} for $L$ if $L$ is a separator for $\zug{L_1,L_2}$ and  $\zug{L_1,L_2}$ is not DBW-separable. Note that the latter indeed implies that $L$ is not in DBW.

A simple no-DBW witness for $L$ can be obtained as follows. Let $\R$ be a DBW refuter for $L$. Then, we define $L_1 = \{\R(y) : y \in  \neg \infty \acc\}$ and $L_2 = \{\R(y) : y \in \infty \acc\}$.  By the definition of DBW-refuters, we have $L_1 \subseteq L$ and $L_2 \cap L =\emptyset$, and so $\zug{L_1,L_2}$ is a no-DBW witness for $L$. It is simple, in the sense that when we describe $L_1$ and $L_2$ by a tree obtained by pruning the $\Sigma^*$-tree, then each node has at most two children -- these that correspond to the responses of $\R$ to $\acc$ and $\rej$. \hfill $\blacksquare$
\end{remark}

\subsection{Refuting Separability}
For a pair of languages $\zug{L_1,L_2}$, we define the language $\DBWsep(L) \subseteq (\Sigma \times A)^\omega$ of words with correct annotations for separation. Thus,
\[\DBWsep(L_1,L_2) =\{x \oplus y : (x \in L_1 \rightarrow y \in \infty{\acc}) \wedge (x \in L_2 \rightarrow y \not \in \infty{\acc}) \}.\]
Note that $\comp (\DBWsep(L_1,L_2))$ is then the language
\[\NoDBWsep(L_1,L_2)=\{x \oplus y : (x \in L_1 \wedge y \not \in \infty{\acc}) \vee (x \in L_2 \wedge y \in \infty{\acc}) \}.\]

A {\em DBW-sep-refuter for $\zug{L_1,L_2}$\/} is an $(A/\Sigma)$-transducer with $\init={\it env}$ that realizes $\NoDBWsep(L_1,L_2)$.

\begin{example}
\label{example separator refuter}
Consider the language $L_{\neg \infty a}=(a+b)^* \cdot b^\omega$, which is not DBW. Let $I=a^* \cdot b^\omega + b^* \cdot a^\omega$, thus we are indifferent about words with only one alternation between $a$ and $b$. In \Cref{sep ref} we describe a DBW-sep refuter for $\zug{L_{\neg \infty a} \setminus I, \comp(L_{\neg \infty a} \cup I)}$. Note that the refuter generates only words in $a \cdot b \cdot a \cdot(a+b)^\omega$, whose intersection with $I$ is empty. Consequently, the refutation is similar to the DBW-refutation of $L_{\neg \infty a}$.
\hfill $\blacksquare$ \end{example}
\begin{figure}[bt]
  \begin{center}
  	\includegraphics[scale=0.8]{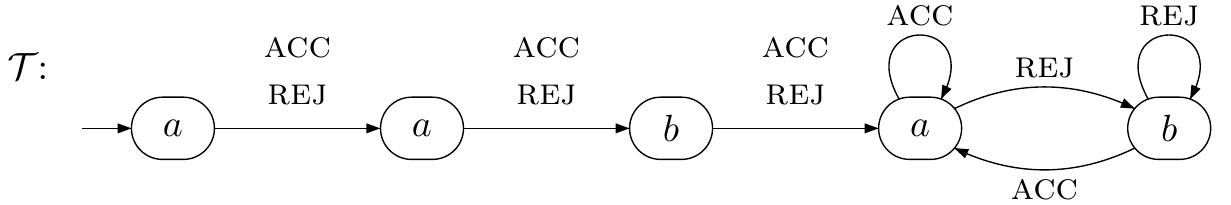}
  \end{center}
  \vspace{-3mm}
  \caption{A DBW-sep refuter for $\zug{L_{\neg \infty a} \setminus I, \comp(L_{\neg \infty a} \cup I)}$.}
  \label{sep ref}
  \vspace{-3mm}
\end{figure}
By \Cref{determinacy}, we have the following extension of \Cref{dbw:recognisability}.
\begin{proposition}\label{dbw:sep recognisability}
Consider two languages $L_1,L_2 \subseteq \Sigma^\omega$. Let $A=\{\acc,\rej\}$. Exactly one of the following holds:
\begin{itemize}
\item
$\zug{L_1,L_2}$ is DBW-separable, in which case the language $\DBWsep(L_1,L_2)$ is $(\Sigma/A)$-realizable by the system, and a finite-memory winning strategy for the system induces a DBW for a language $L$ that separates $L_1$ and $L_2$.
\item
$\zug{L_1,L_2}$ is not DBW-separable, in which case the language $\NoDBWsep(L)$ is $(A/\Sigma)$-realizable by the environment, and a finite-memory winning strategy for the environment induces a DBW-sep-refuter for $\zug{L_1,L_2}$.
\end{itemize}
\end{proposition}

As for complexity, the construction of the game for $\DBWsep(L_1,L_2)$ is similar to the one described in \Cref{dbw by dpw yes or refute}. Here, however, the input to the problem includes two DPWs. Also, the positive case, namely the construction of the separator does not follow from known results.

\begin{theorem}
\label{dbw-sep by dpw yes or refute}
Consider DPWs $\A_1$ and $\A_2$ with $n_1$ and $n_2$ states, respectively. Let $L_1=L(\A_1)$ and $L_2=L(\A_2)$. One of the following holds.
\begin{enumerate}
\item
There is a DBW $\A$ with $2 \cdot n_1 \cdot n_2$ states such that $L(\A)$ DBW-separates $\zug{L_1,L_2}$.
\item
There is a DBW-sep-refuter for $\zug{L_1,L_2}$ with $2 \cdot n_1 \cdot n_2$ states.
\end{enumerate}
\end{theorem}

\begin{proof}
We show that $\DBWsep(L_1,L_2)$ and $\NoDBWsep(L_1,L_2)$ can be recognised by DRWs with at most $2 \cdot n_1 \cdot n_2$ states. Then, by \cite{EJ88}, we can construct a DBW or a DBW-sep-refuter with at most $2 \cdot n_1 \cdot n_2$ states. The construction is similar to the one described in the proof of \Cref{dbw by dpw yes or refute}. The only technical challenge is the fact $\DBWsep(L_1,L_2)$ is defined as the intersection, rather than union, of two languages. For this, we observe that we can define $\DBWsep(L_1,L_2)$ also as $\{x \oplus y : (y \in \infty \acc \text{ and } x \notin L_2) \text{ or } (y \notin \infty \acc \text{ and } x \notin L_1)\}$. With this formulation we then can reuse the union construction as seen in \Cref{dbw by dpw yes or refute} to obtain DRWs with at most $2 \cdot n_1 \cdot n_2$ states. \hfill \qed
\end{proof}

As has been the case with DBW-recognizability, one can generate certificates from a DBW-sep-refuter. The proof is similar to that of \Cref{lem:not-dbw-cert}, with membership in $L_1$ replacing membership in $L$ and membership in $L_2$ replacing being disjoint from $L$. Formally, we have the following.
\begin{theorem}\label{lem:not-dbw-sep-cert}
Two $\omega$-regular languages $L_1, L_2 \subseteq \Sigma^\omega$ are not DBW-separable iff there exist three finite words $x \in \Sigma^*$ and $x_1, x_2 \in \Sigma^+$, such that
\ifarxiv
\[x \cdot (x_1 + x_2)^* \cdot x_1^\omega \subseteq L_1 \quad \text{ and } \quad x \cdot (x_1^* \cdot x_2)^\omega \subseteq L_2.\]
\fi
\end{theorem}
We refer to a triple $\zug{x,x_1,x_2}$ of words that satisfy the conditions in \Cref{lem:not-dbw-sep-cert} as a {\em certificate\/} to the non-DBW-separability of $\zug{L_1, L_2}$.
Observe that the same way we generated a no-DBW witness in \Cref{core lang}, we can extract, given a DBW-sep-refuter $\R$ for $\zug{L_1, L_2}$, languages  $L'_1 \subseteq L_1$ and $L'_2 \subseteq L_2$ that tighten $\zug{L_1,L_2}$ and are still not DBW-separable.

\subsection{Certificate-Guided Approximation}
In this section we describe a method for finding small approximating languages $I_\downarrow$ and $I_\uparrow$ such that $\zug{L \setminus I_\downarrow, \comp(L) \setminus I_\uparrow}$ is DBW-separable. If this method terminates we obtain an approximation for $L$ that is DBW-recognizable. As in {\em counterexample guided abstraction-refinement\/} (CEGAR) for model checking \cite{CGJLV03}, we use certificates for non-DBW-separability in order to suggest interesting approximating languages. Intuitively, while in CEGAR the refined system excludes the counterexample, here the approximation of $L$ excludes the certificate.

Consider a certificate $\zug{x, x_1, x_2}$ for the non-DBW-separability of $\zug{L_1, L_2}$. We suggest the following five approximations:
\[\begin{array}{lcl}
C_0 = x \cdot (x_1 + x_2)^\omega                                  & \rightsquigarrow \ & \zug{L_1 \setminus C_0, L_2 \setminus C_0} \\
C_1 = x \cdot (x_1 + x_2)^* \cdot x_1^\omega = L_1 \cap C_0 \quad & \rightsquigarrow \ & \zug{L_1 \setminus C_1, L_2} \\
C_2 = x \cdot (x_2^* \cdot x_1)^\omega \supset C_1                & \rightsquigarrow \ & \zug{L_1, L_2 \setminus C_2} \\
C_3 = x \cdot (x_1^* \cdot x_2)^\omega = L_2 \cap C_0             & \rightsquigarrow \ & \zug{L_1, L_2 \setminus C_3} \\
C_4 = x \cdot (x_1 + x_2)^* \cdot x_2^\omega \subset C_3          & \rightsquigarrow \ & \zug{L_1, L_2 \setminus C_4}
\end{array}\]

First, it is easy to verify that $\zug{x, x_1, x_2}$ is indeed not a certificate for the non-DBW-separability of the obtained candidate pairs $\zug{L_1',L_2'}$.
If $\zug{L_1',L_2'}$ is DBW-separable, we are done (yet may try to tighten the approximation). Otherwise, we can repeat the process with a certificate for the non-DBW-separability of $\zug{L_1',L_2'}$. As in CEGAR, some suggestions may be more interesting than others, in some cases the process terminates, in some it does not, and the user takes part directing the search.

\begin{example}
Consider again the language $L = (a+b)^* \cdot b^\omega$ and the certificate $\zug{x, x_1, x_2} = \zug{\epsilon, b, a}$. Trying to approximate $L$ by a language in DBW, we start with the pair $\zug{L,\comp(L)}$. Our five suggestions are then as follows.
\[\begin{array}{lcl}
C_0 = \Sigma^\omega                                               & \rightsquigarrow \ & \zug{L \setminus C_0, \comp(L) \setminus C_0} = \zug{\emptyset, \emptyset} \\
C_1 = (b+a)^* \cdot b^\omega  \quad                               & \rightsquigarrow \ & \zug{L \setminus C_1, \comp(L)} = \zug{\emptyset, \comp(L)} \\
C_2 = (a^* \cdot b)^\omega                                        & \rightsquigarrow \ & \zug{L, \comp(L) \setminus C_2} = \zug{L, (a+b)^* \cdot a^\omega} \\
C_3 = (b^* \cdot a)^\omega 							              & \rightsquigarrow \ & \zug{L, \comp(L) \setminus C_3} = \zug{L, \emptyset} \\
C_4 = (b+a)^* \cdot a^\omega                                      & \rightsquigarrow \ & \zug{L, \comp(L) \setminus C_4} = \zug{L, (a+b)^* \cdot (a \cdot a^* \cdot b \cdot b^*)^\omega}
\end{array}\]

Candidates $C_0$, $C_1$, and $C_3$ induce trivial approximations. Then, $C_2$ suggests to over-approximate $L$ by setting $I_\uparrow$ to $(a^* \cdot b)^\omega$, which we view as a nice solution, approximating ``eventually always $b$" by ``infinitely often $b$". Then, the pair derived from $C_4$ is not DBW-separable. We can try to approximate it. Note, however, that repeated approximations in the spirit of $C_4$ are going to only extend the prefix of $x$ in the certificates, and the process does not terminate.
\ifarxiv

Let us now consider the slightly different certificate $\zug{x, x_1, x_2} =  \zug{a, b, a}$ and the derived candidates:
\[\begin{array}{lcl}
C_0 = a \cdot \Sigma^\omega                                               & \rightsquigarrow & \zug{L \setminus C_0, \comp(L) \setminus C_0} = \zug{b \cdot L, b \cdot \comp(L)} \\
C_1 = a \cdot (b+a)^* \cdot b^\omega  \quad                               & \rightsquigarrow & \zug{L \setminus C_1, \comp(L)} = \zug{b \cdot L, \comp(L)} \\
C_2 = a \cdot (a^* \cdot b)^\omega                                        & \rightsquigarrow & \zug{L, \comp(L) \setminus C_2} \\ & & = \zug{L, b \cdot \comp(L) + a \cdot (a+b)^* \cdot a^\omega} \\
C_3 = a \cdot (b^* \cdot a)^\omega 							              & \rightsquigarrow & \zug{L, \comp(L) \setminus C_3} = \zug{L, b \cdot \comp(L)} \\
C_4 = a \cdot (b+a)^* \cdot a^\omega                                      & \rightsquigarrow & \zug{L, \comp(L) \setminus C_4} \\ & & = \zug{L, b \cdot \comp(L) + a \cdot (a+b)^* \cdot (a \cdot a^* \cdot b \cdot b^*)^\omega}
\end{array}\]

One can easily verify that $\zug{x, x_1, x_2} = \zug{b \cdot a, b, a}$ is a certificate showing that none of the suggested pairs are DBW-separable. In fact $\zug{x, x_1, x_2} = \zug{b^i \cdot a, b, a}$, for $i = 0,1,2,\dots$, describes an infinite sequence such that no refinement obtained after a finite number of steps is DBW-separable.
\fi \hfill $\blacksquare$
\end{example}

\section{Other Classes of Deterministic Automata}\label{sec:other}

In this section we generalise the idea of DBW-refuters to other classes of deterministic automata.
For this we take again the view that a deterministic automaton is a $\zug{\Sigma, A}$-transducer over a suitable annotation alphabet $A$.
We then characterize each class of deterministic automata by two languages over $A$:
\begin{itemize}
\item
The language $\Lacc \subseteq A^\omega$, describing when a run is accepting. For example, for DBWs, we have $A=\{\acc,\rej\}$ and $\Lacc = \infty \acc$.
\item
The language $\Llegal \subseteq A^\omega$, describing structural conditions on the run. For example, recall that a DWW is a DBW in which the states of each SCS are either all accepting or all rejecting, and so each run eventually get trapped in an accepting or rejecting SCS. Accordingly, the language of runs that satisfy the structural condition is $\Llegal= A^* \cdot (\acc^\omega + \rej^\omega)$.
\end{itemize}

We now formalize this intuition. Let $A$ be a finite set of annotations and let $\gamma=\zug{\Lacc, \Llegal}$, for $\Lacc, \Llegal \subseteq A^\omega$.
A  deterministic automaton $\A = \zug{\Sigma, Q, q_0, \delta,\alpha}$ is a deterministic $\gamma$ automaton (\dgamw, for short) if there is a function $\tau \colon Q \to A$ that maps each state to an annotation such that a run $r$ of $\A$ satisfies $\alpha$ iff $\tau(r) \in \Lacc$, and all runs $r$ satisfy the structural condition, thus  $\tau(r) \in \Llegal$.
We then say that a language $L$ is $\gamma$-recognizable if there a \dgamw \/ $\A$ such that $L = L(\A)$.

Before we continue to study $\gamma$-recognizability, let us demonstrate the $\gamma$-characterization of common deterministic automata.
We first start with classes $\gamma$ for which $\Llegal$ is trivial; i.e., $\Llegal = A^\omega$.
\begin{itemize}
    \item DBW: $A=\{\acc,\rej\}$ and $\Lacc =\infty \acc$.
    \item DCW: $A=\{\acc,\rej\}$ and $\Lacc =\neg \infty \acc$.
	\item DPW[$i,k$]: $A = \{i, \dots, k\}$ and $\Lacc= \{ y \in A^\omega : \max(\infi(y)) \text{ is odd}\}$.
	\item DELW[$\theta$]: $A=2^\mathbb{M}$ and $\Lacc = \{ y \in A^\omega : y \models \theta\}$.
\end{itemize}
Note that the characterizations for B\"uchi, co-B\"uchi, and parity are special cases of the characterization for DELW. In a similar way, we could define a language $\Lacc$ for DRW[$k$] and other common special cases of DELWs.
We continue to classes in the depth hierarchy, where $\gamma$ includes also a structural restriction:
\begin{itemize}
	\item DWW: The set $A$ and the language $\Lacc$ are as for DBW or DCW. In addition, $\Llegal=A^* \cdot (\acc^\omega + \rej^\omega)$.
	\item DWW[$j,k$], for  $j \in \{0,1\}$: The set $A$ and the language $\Lacc$ are  as for DPW[$j,k$]. In addition, $\Llegal = \{ y_0 \cdot y_1 \cdots \in A^\omega:$ for all $i \geq 0$, we have that $y_i \leq y_{i+1}\}$.
	\item {\em Bounded Languages}: A language $L$ is bounded if it is both safety and co-safety. Thus, every word $w \in \Sigma^\omega$ has a prefix $v \in \Sigma^*$ such that either for all $u \in \Sigma^\omega$ we have $v \cdot u \in L$, or for all $u \in \Sigma^\omega$ we have $v \cdot u \not \in L$ \cite{KV01b}. To capture this, we use $A = \{\acc, \rej, ?\}$, where $``?"$ is used for annotating states with both accepting and rejecting continuations. Then, $\Lacc = A^* \cdot \acc^\omega$, and $\Llegal = ?^* \cdot (\acc^\omega + \rej^\omega)$.
		\item
	{\em Deterministic $(m,n)$-Superparity Automata\/} \cite{PP04}: $A = \{(i,j) : 0 \leq i \leq m, 0 \leq j \leq n\}$, $\Lacc = \{y_m \oplus y_n \in A^\omega : \max(\infi(y_m)) + \max(y_n) \text{ is odd}\}$, and $\Llegal = \{ y_m \oplus (y_0 \cdot y_1 \cdots) \in A^\omega :  y_i \leq y_{i+1}, \mbox{ for all } i \geq 0\}$.
	\end{itemize}

Let $\Sigma$ be an alphabet, let $A$ be an annotation alphabet, and let $\gamma = \zug{\Lacc,$ $\Llegal}$, for $\Lacc, \Llegal \subseteq A^\omega$. We define the language $\Real(L, \gamma) \subseteq (\Sigma \times A)^\omega$ of words with correct annotations.
\[\Real(L, \gamma)       = \{x  \oplus y : y \in \Llegal \text{ and } (x \in L \text{ iff } y \in \Lacc)\}.\]
Note that the language $\DBWreal(L)$ can be viewed as a special case of our general framework. In particular, in cases $\Llegal=A^\omega$, we can remove the $y \in \Llegal$ conjunct from
$\Real(L, \gamma)$.
Note that $\comp (\Real(L,\gamma))$ is the language
\[\NoReal(L,\gamma)=\{x \oplus y : y \not \in \Llegal \mbox{ or } (x \in L \mbox{ iff } y \not \in \Lacc)\}.\]
A {\em $\gamma$-refuter for $L$\/} is then an $(A/\Sigma)$-transducer with $\init={\it env}$ that realizes $\NoReal(L,\gamma)$.
We can now state the ``\dgamw-generalization" of Proposition~\ref{dbw:recognisability}.

\begin{proposition}\label{general:recognisability}
Consider an $\omega$-regular language $L \subseteq \Sigma^\omega$, and a pair $\gamma=\zug{\Lacc, \Llegal}$, for $\omega$-regular languages $\Lacc, \Llegal \subseteq A^\omega$. Exactly one of the following holds:
\begin{enumerate}
		\item
		$L$ is in \textnormal{\dgamw}, in which case the language $\Real(L,\gamma)$ is $(\Sigma/A)$-realizable by the system, and a finite-memory winning strategy for the system induces a \textnormal{\dgamw} for $L$.
		\item
		$L$ is not in \textnormal{\dgamw}, in which case the language $\NoReal(L, \gamma)$ is $(A/\Sigma)$-re\-alizable by the environment,  and a finite-memory winning strategy for the environment induces a $\gamma$-refuter for $L$.
\end{enumerate}
\end{proposition}

Note that every DELW can be complemented by dualization, thus by changing its acceptance condition from $\theta$ to $\neg \theta$. In particular, DBW and DCW dualize each other. As we argue below, dualization is carried over to refutation. For example, the $(\{\acc,\rej\}/\Sigma)$-transducer $\R$ from \Cref{ref fg1} is both a DBW-refuter for $\neg \infty a$ and a DCW-refuter for $\infty a$. Formally, we have the following.
\begin{theorem}
\label{cer dual}
Consider an EL-condition $\theta$ over $\mathbb{M}$. Let $A=2^\mathbb{M}$. For every $(A/\Sigma)$-transducer $\R$ and language $L$, we have that $\R$ is a \textnormal{DELW}$[\theta]$-refuter for $L$ iff $\R$ is a \textnormal{DELW}$[\neg \theta]$-refuter for $\comp(L)$. In particular,
for every language $L$ and $(\{\acc,\rej\}/\Sigma)$-transducer $\R$, we have that $\R$ is a DBW-refuter for $L$ iff $\R$ is a DCW-refuter for $\comp(L)$.
\end{theorem}

\begin{proof}
For DELW[$\theta$]-recognizability of $L$, the language of correct annotations is $\{x  \oplus y :(x \in L \text{ iff } y \models \theta)\}$, which is equal to $\{x  \oplus y :(x \in \comp(L) \text{ iff } y \models \neg \theta)\}$, which is the language of correct annotations for DELW[$\neg \theta$]-recognizability of $\comp(L)$. \hfill \qed
\end{proof}

While dualization is nicely carried over to refutation, this is not the case for all expressiveness results. For example, while DWW=DBW$\cap$DCW, and in fact DBW and DCW are weak type (that is, when the language of a DBW is in DWW, an equivalent DWW can be defined on top of its structure, and similarly for DCW \cite{KMM06}), we describe \ifarxiv below \fi a DWW-refuter that is neither a DBW- nor a DCW-refuter. Intuitively, this is possible as in DWW refutation, Prover loses when the input is not in  $A^* \cdot (\acc^\omega + \rej^\omega)$, whereas in DBW and DCW refutation, Refuter has to respond correctly also for these inputs.

\ifarxiv
\begin{example}
Let $\Sigma=\{a,b,c,d\}$, and $A=\{\acc,\rej\}$.
Consider the language $L=(a^+ \cdot b \cdot c^* \cdot d)^* \cdot a^\omega + (a \cdot b \cdot d)^\omega$.
Note that $L$ is in DCW, but not in DBW, and hence also not in DWW. The $(A/\Sigma)$-transducer $\R$ in Figure~\ref{dwwref} is a DWW-refuter for $L$.
To see this, recall that for DWWs, we have that $\Llegal=A^* \cdot (\acc^\omega + \rej^\omega)$, and so all input sequences $y \in A^\omega$ that satisfy $\Llegal$ eventually gets trapped in the $a^\omega$ loop, generating a rejecting run on a word in the language, or gets trapped in the $c^\omega$ loop, generating an accepting run on a word not in the language.

On the other hand, while $L$ is not in DBW, the transducer $\R$ is not a DBW-refuter for $L$. To see this, observe that the DBW $\A$ in the figure suggests a winning strategy for Prover in the game corresponding to DBW. Indeed, when Prover generates $(\rej \cdot \acc \cdot \rej)^\omega$, which is accepting, then by following $\R$, Refuter responds with $(a \cdot b \cdot d)^\omega$, which is in $L$, and so Prover wins. Note that, unsurprisingly, the input generated by Prover does not satisfy $\Llegal$. \hfill $\blacksquare$

\begin{figure}[t]
\begin{center}
\includegraphics[width=\textwidth]{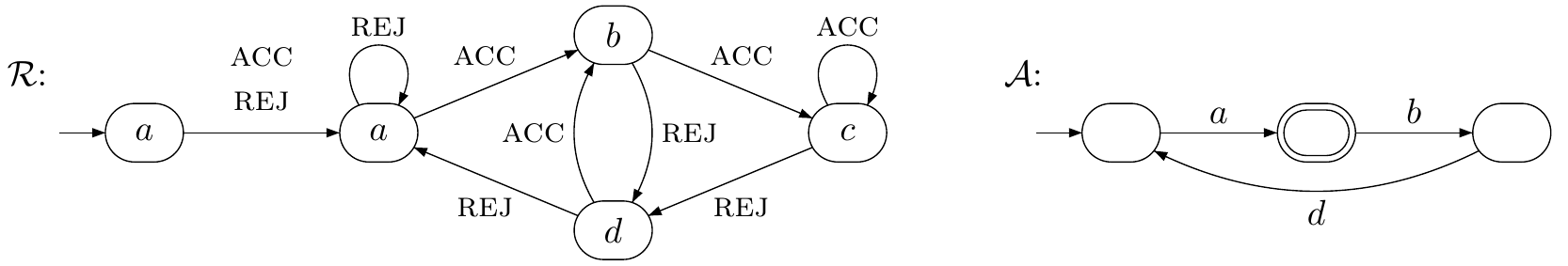}
\caption{The DWW-refuter $\R$ looses as a DBW-refuter when it plays against $\A$.}
\label{dwwref}
\end{center}
\end{figure} \end{example}
\fi

On the other hand, as every DWW is also a DBW and a DCW, every DBW-refuter or DCW-refuter is also a DWW-refuter.

\ifarxiv
\paragraph{Separability and Approximation.}
Consider a characterization $\gamma=\zug{\Lacc, \Llegal}$. Two languages $L_1, L_2 \subseteq \Sigma^\omega$ are \emph{$\gamma$-separable} if there exists a \dgamw $\A$ such that $L_1 \subseteq L(\A)$ and $L_2 \cap L(\A) = \emptyset$.
We define the corresponding languages of correct and incorrect annotations as follows.
\begin{itemize}
\item
$\Sep(L_1, L_2, \Lacc, \Llegal) =  \\
\{x  \oplus y : y \in \Llegal \text{ and }  ((x \in L_1 \text{ and } y \in \Lacc) \text{ or } (x \in L_2 \text{ and } y \notin \Lacc))\}$.
\item
$\NoSep(L_1, L_2, \Lacc, \Llegal) =  \comp(\Sep(L_1, L_2, \Lacc, \Llegal)) = \\
                                 \{x  \oplus y : y \notin \Llegal \text{ or }
	                               ((x \in L_1 \text{ and } y \notin \Lacc) \text{ or } (x \in L_2 \text{ and } y \in \Lacc))\}$.
\end{itemize}

Note that the language $\DBWsep(L_1, L_2)$ can be viewed as a special case of our general framework and as before in cases $\Llegal=A^\omega$, we can remove the $y \in \Llegal$ conjunct from $\Sep$. A {\em $\gamma$-sep-refuter for $L$\/} is an $(A/\Sigma)$-transducer with $\init={\it env}$ that realizes $\NoSep(L_1, L_2, \Lacc, \Llegal)$. By \Cref{determinacy}, exactly one of the following holds:

\begin{proposition}\label{general:separability}
Consider $\omega$-regular languages $L_1, L_2 \subseteq \Sigma^\omega$, and a characterization $\gamma=\zug{\Lacc, \Llegal}$, for $\omega$-regular languages $\Lacc, \Llegal \subseteq A^\omega$. Exactly one of the following holds:
\begin{enumerate}
		\item
		$\zug{L_1,L_2}$ are $\gamma$-separable, in which case the language $\Sep(L_1, L_2,\gamma)$ is $(\Sigma/A)$-realizable by the system, and a finite-memory winning strategy for the system induces a \textnormal{\dgamw} for some $L$ such that $L_1 \subseteq L$ and $L \cap L_2 = \emptyset$.
		\item
		$\zug{L_1,L_2}$ are not $\gamma$-separable, in which case the language $\NoSep(L_1, L_2, \gamma)$ is $(A/\Sigma)$-realizable by the environment, and a finite-memory winning strategy for the environment induces a $\gamma$-sep-refuter for $\zug{L_1,L_2}$.
\end{enumerate}
\end{proposition}
\fi

\ifarxiv

\section{Certifying \dgamw-Refutation}

In this section we extend the three-word certificates for non-DBW-recognizability to richer classes of deterministic automata. The idea is similar (and in fact a little tedious): each \dgamw-refuter embodies a structure (analogous to the one in Lemma~\ref{lem: lan struct}) from which we can extract finite words that constitute the corresponding certificate (analogous to the one in Theorem~\ref{lem:not-dbw-cert}). We describe here the details for classes in the Mostowski hierarchy and well as for classes of the depth-hierarchy. We also restrict ourselves to word-certificates for non-recognizability and do not show the word-certificates for non-separability which have an identical structure.

\subsection{Mostowski Hierarchy}

First, by Theorem~\ref{cer dual}, certificates for a class and its dual class are related. For example, dualizing Theorem~\ref{lem:not-dbw-cert}, we obtain certificates for non-DCW-recogniz\-ability as follows.
\begin{theorem}\label{lem:not-dcw-cert}
An $\omega$-regular language $L$ is not in DCW iff there exist three finite words $x \in \Sigma^*$ and $x_1, x_2 \in \Sigma^+$, such that
\[x \cdot (x_1 + x_2)^* \cdot x_1^\omega \cap L = \emptyset \qquad \text{ and } \qquad  x \cdot (x_1^* \cdot x_2)^\omega \subseteq L.\]
\end{theorem}

Handling DPWs, we first define the analogue of a $\rej^+$-path, and then point to the desired structure and the certificate it induces.
Consider a DPW[$i,k$]-refuter $\R=\zug{\{i, \dots, k\},\Sigma,{\it env},S,s_0,\rho,\tau}$ with $i \in \{0,1\}$ and $i \leq k$. Let $\ell \in \{i, \dots, k\}$. We say that a path $s_1,\ldots,s_m$ in $\R$ is an {\em $\ell^+_\leq$-path\/} if its first transition is labelled $\ell$ and all its other transitions are labeled by colors in $\{i, \dots \ell\}$. Thus, $s_{2}=\rho(s_1, \ell)$ and, for all $1 \leq j < m$, we have that $s_{j+1}=\rho(s_j, \ell')$, for some $\ell' \leq \ell$.

\begin{lemma}\label{lem:flower}
Consider a \textnormal{DPW}$[i,k]$-refuter $\R=\zug{\{i, \dots, k\},\Sigma,{\it env},S,s_0,\rho,\tau}$ with $i \in \{0,1\}$ and $i \leq k$. There exists a state $s \in S$, a (possibly empty) path $p = s_0,s_1, \dots s_m$, and for each $\ell \in \{i, \dots, k\}$,  a $\ell^+_\leq$-cycle $p_\ell = s^\ell_{1} \dots s^\ell_{m_\ell}$, such that $s_m=s^\ell_1=s^\ell_{m_\ell}=s$.
\end{lemma}

\begin{proof}
Let $\R_{\leq j}$ denote the transducer that we obtain from $\R$ when we restrict $\delta$ to transitions labelled by at most $j$. Note that $\R$ is $\R_{\leq k}$. We proceed by induction on $j$ with $i \leq j \leq k$ and show that in the transducer $\R_{\leq j}$ for every state $s \in S$ there exists a state $s' \in S$, a (possibly empty) path $p = s_1, \dots s_m$ with $s = s_1$, and that for each $\ell \in \{i, \dots, j\}$ there exists a $\ell^+_\leq$-cycle $p_\ell = s^\ell_1,s^\ell_{2} \dots s^\ell_{m_\ell}$, such that $s_m=s^\ell_1=s^\ell_{m_\ell}=s'$. The base case for $j = i$ follows immediately from the fact that $\R_{\leq i}$ is responsive on $\{i\}$ and by reading $i^\omega$ we obtain a lasso with the required properties.

Let $j > i$ and let $s \in S$ be an arbitrary state. Further, let $s_{j} \in S$ be a reachable state from $s$ that belongs to an ergodic component in the graph of $\R_{\leq j}$ (that is, $s_{j} \in C$, for a set $C$ of strongly connected states that can reach only states in $C$). By induction hypothesis there exists $s' \in S$, a (possibly empty) path $p = s_j,s_{j+1}, \dots s_m$, and for each $\ell \in \{i, \dots, j-1\}$ there exists a $\ell^+_\leq$-cycle $p_\ell = s^\ell_1,s^\ell_{2} \dots s^\ell_{m_\ell}$, such that $s_m=s^\ell_1=s^\ell_{m_\ell}=s'$ for every $\ell \in \{i, \dots, j-1\}$. Since $\R_{\leq j}$ is responsive on $\{i, \dots, j\}$ we can take from $s'$ a transition labelled $\ell$ and since $C$ is ergodic we can find a path back to $s'$. Thus we obtain the missing $j^+_{\leq}$-cycle and by concatenating the path from $s$ to $s_j$ and the path $p$, we show that $s'$ can be reached from $s$.
\hfill \qed \end{proof}

\begin{theorem}\label{lem:not-dpw-0-cert}
Let $i \in \{0,1\}$ and $i \leq k$. An $\omega$-regular language $L$ is not in \textnormal{DPW}$[i,k]$ iff there exist finite words $x \in \Sigma^*$ and $x_i, \dots, x_k \in \Sigma^+$, such that for every even $i \leq \ell \leq k$, we have
\[x \cdot (x_i + \dots + x_k)^* \cdot ((x_{i} + x_{i+1} + \dots + x_{\ell - 1})^* \cdot x_\ell)^\omega \subseteq L,\]
and for every odd $i \leq \ell \leq k$, we have
\[x \cdot (x_i + \dots + x_k)^* \cdot ((x_{i} + x_{i+1} + \dots + x_{\ell - 1})^* \cdot x_\ell)^\omega \cap L = \emptyset.\]
\end{theorem}

\begin{proof}
	Assume first that $L$ is not in DPW[$i,k$]. Then, by \Cref{general:recognisability}, there exists a DPW[$i,k$]-refuter $\R$ for it. From this refuter we can extract via \Cref{lem:flower} a path $p$ and $\ell^+_\leq$-cycles. We then construct the postulated finite words in the exact same way as in the proof of \Cref{lem:not-dbw-cert}.


For the other direction, we first simplify the presentation by assuming $i = 0$. The proof for $i = 1$ is analogous. Assume by way of contradiction that there is a DPW$[0,k]$ $\A$ with $\lang(\A)=L$.
Let $\A=\zug{\Sigma,Q,q_0,\delta,\alpha}$.
Let $n = |Q|$ and consider the following sequence of words $w_0 = x_0^n$, $w_1 = (w_0 \cdot x_1)^n$, \dots, $w_k = (w_{k-1} \cdot x_k)^n$.
Let $q = \delta(q_0, w)$ be a state that is reached after reading $w \in x \cdot (x_i + x_{i+1} + \dots x_k)^*$. Since $w \cdot w_0^\omega \in L$, there must be a state $p_0$ that is visited infinitely often and $\alpha(p_0)$ is odd. Since $|w_0| \geq |Q|$, this state must have been visited while reading $w_0$. Now, consider $w \cdot w_1^\omega$. This word is rejected and by the same reasoning as before there must be some $p_1$ such that $\alpha(p_1)$ is even, it is visited while reading $w_1$, and for every $p_0$ that belongs to a $w_0$ subsequences we have $\alpha(p_1) > \alpha(p_0)$. We continue and obtain a sequence $\alpha(p_k) > \dots > \alpha(p_0)$ with $k$ strict inequalities. Since $\alpha(p_0)$ is odd, we have $\alpha(p_0) > 0$ and thus $\alpha(p_k) > k$, which contradicts the fact that $\A$ is a DPW$[0,k]$.
\hfill \qed \end{proof}

Note that, by \cite{NW98}, the ``flower''-structure that induces the certificate exists also in DPWs for $L$. Specifically, while
\Cref{lem:flower} shows that every DPW$[i,k]$-refuter contains a ``flower'' with $k-i+1$ petals, it is shown in \cite{NW98} that for every $\omega$-language $L$ not in DPW[$1,k+1$], there exists a DPW for $L$ that contains a flower with $k+1$ petals and this flower occurs in some accepting run.

\paragraph{Rabin and Streett acceptance.} Recall that for all $k \geq 0$, we have that DRW[$k$] = DPW[$0,2k$] . Hence, the certificates obtained through \Cref{lem:not-dpw-0-cert} carry over to the Rabin case. Further, in a deterministic generalized Rabin automaton (DGRW), the acceptance condition is of the form \[\alpha = \{\zug{B_1, G_{1,1}, \dots, G_{1,n_1}}, \dots \zug{B_k, G_{k,1}, \dots, G_{n,k_n}}\},\] and a run $r$ is accepting if there is $j \in \{1,\dots,k\}$,  such that $\infi(r) \cap B_j = \emptyset$ and $\infi(r) \cap G_{j, \ell} \neq \emptyset$ for every $1 \leq \ell \leq n_j$. Since degeneralization does not increase the number of Rabin pairs, we have that  DGRW[$k$] = DRW[$k$] = DPW[$0,2k$], and so again the certificates obtained through \Cref{lem:not-dpw-0-cert} are applicable. Nevertheless, a refuter for the DRW[$k$] may be more succinct than a DPW$[0,2k]$-refuter.

Finally, the Streett and generalized acceptance conditions are dual to Rabin and generalized Rabin, and certificates for them can be obtained dually.

\subsection{Depth-Hierarchy}

We continue to certificates for non-DWW[$i,k$]-recognizability.
Consider a DWW[$i,k$]-refuter $\R=\zug{\{i, \dots, k\},\Sigma,{\it env},S,s_0,\rho,\tau}$, with $i \in \{0,1\}$ and $i \leq k$. Let $\ell \in \{i, \dots, k\}$. We say that a path $s_1,\ldots,s_m$ in $\R$ is an $\ell^+$-path if all transitions are labelled by $\ell$. Thus, for all $1 \leq j < m$, we have that $s_{j+1}=\rho(s_j, \ell)$.

\begin{lemma}\label{lassosequence}
Consider a \textnormal{DWW}$[i,k]$-refuter $\R=\zug{\{i, \dots, k\},\Sigma,{\it env},S,s_0,\rho,\tau}$ with $i \in \{0,1\}$ and $i \leq k$. Let $s^{i-1}$ be an alias for $s_0$. Then there exists a sequence of states $s^i, s^{i+1}, \dots s^k \in S$, such that for every $j \in \{i, \dots, k\}$ there exists a (possibly empty) $j^+$-path $p^j = s^j_1,s^j_2, \dots s^j_{m_j}$, and a $j^+$-cycle $c^j = s^j_{m_j+1},s^j_{m_j+2} \dots s^j_{m_j+m'_j}$ such that $s^j_{m_j} = s^j_{m_j+1} = s^j_{m_j + m_j'} = s^j$ and $s^j_1 = s^{j-1}$.
\end{lemma}

\begin{proof}
Such a structure can be found by constructing a sequence of lassos. Start by reading $i^\omega$ from $s_0$ to construct an $i^+$-path $p^i$ and an $i^+$-cycle $c^i$. $s^i$ is then the last state of  $c^i$, respectively. Then, continue by reading $(i+1)^\omega$ from $s^i$ to find the next lasso and continue until all lassos are found.
\hfill \qed \end{proof}

\begin{theorem}\label{thm:not-dww-i-j-cert}
Let $i \in \{0,1\}$ and $i \leq k$. An $\omega$-regular language $L$ is not in \textnormal{DWW}$[i,k]$ iff there exist finite words $\hat{x}_{i}, \hat{x}_{i+1}, \dots, \hat{x}_k \in \Sigma^*$ and $x_i, x_{i+1}, \dots, x_k \in \Sigma^+$, such that for every even $i \leq \ell \leq k$, we have
\[\hat{x}_i \cdot x_i^* \cdot \hat{x}_{i+1} \cdot x_{i+1}^* \cdots \hat{x}_{\ell} \cdot x_{\ell}^\omega \subseteq L,\]
and for every odd $i \leq \ell \leq k$, we have
\[\hat{x}_i \cdot x_i^* \cdot \hat{x}_{i+1} \cdot x_{i+1}^* \cdots \hat{x}_{\ell } \cdot x_{\ell }^\omega \cap L = \emptyset.\]
\end{theorem}

\begin{proof}
Assume first that $L$ is not in DWW[$i,k$]. Then, by \Cref{general:recognisability}, there exists a DWW[$i,k$]-refuter $\R$ for it. From this refuter we can extract via \Cref{lassosequence} a sequence of states with the corresponding paths and cycles. We then obtain words in the same manner as in the proof of \Cref{lem:not-dbw-cert}.

For the remaining direction assume by way of contradiction that there is a DWW[$i,k$] $\A =\zug{\Sigma,Q,q_0,\delta,\alpha}$ with $\lang(\A)=L$.
We simplify the presentation by assuming $i = 0$. The proof for $i = 1$ is analogous.
Let $n = |Q|$ and consider the following sequence of words $w_0 = \hat{x}_0 \cdot x_0^n$, $w_1 = w_0 \cdot \hat{x}_1 \cdot x_1^n$, \dots, $w_k = w_{k-1} \cdot \hat{x}_k \cdot x_k^n$. Since $w_0 \cdot x_0^\omega \in L$ and $w_0$ has more letters than $\A$ has states, we have $\alpha(\delta(q_0, w_0))$ is odd. By the same argument we have due to $w_1 \cdot x_1^\omega \notin L$ that $\alpha(\delta(q_0, w_1))$ is even and since $w_0$ is a prefix of $w_1$ we also have $\alpha(\delta(q_0, w_1)) > \alpha(\delta(q_0, w_0))$. Continuing in this manner we obtain a chain of length $\alpha(\delta(q_0, w_k)) > \alpha(\delta(q_0, w_{k-1})) > \dots > \alpha(\delta(q_0, w_0))$ with $k$ strict inequalities. Since the smallest element is odd, we have $\alpha(\delta(q_0, w_0)) > 0$ and thus $\alpha(\delta(q_0, w_k)) > k$ which contradicts $\A$ being a DWW$[0,k]$.
\hfill \qed \end{proof}

We continue with general DWWs.

\begin{lemma}\label{twosnakes}
Consider a \textnormal{DWW}-refuter $\R=\zug{\{\acc, \rej\},\Sigma,{\it env},S,s_0,\rho,\tau}$. There exist two states $s^1, s^2 \in S$, (possibly empty) paths $p_0 = s_0,s_1, \dots s_{m_0}$, $p_1 = s_{m_0+1},\dots,s_{m_0 + m_1}$, and $p_2 = s_{m_0 + m_1 + 1},\dots,s_{m_0 + m_1 + m_2}$, a $\rej^+$-cycle $c^1 = s^1_1,s^1_{2} \dots s^1_{l_1}$, and a $\acc^+$-cycle $c^{2} = s^{2}_1,s^{2}_{2} \dots s^{2}_{l_2}$, such that $s_{m_0} = s_{m_0 + 1} = s_{m_0 + m_1 + m_2} = s^{1}_1 = s^{1}_{l_1}$ and $s_{m_0 + m_1} = s_{m_0 + m_1 + 1} = s^{2}_1 = s^{2}_{l_2}$.
\end{lemma}

\begin{proof}
 Let $s \in S$ be state in an ergodic SCC of the graph of $\R$. Then the $\acc^+$- and $\rej^+$-cycle are obtained from the lassos formed by reading from $s$ the words $\acc^\omega$ and $\rej^\omega$, respectively. Since $s$ belongs to an ergodic SCC, there exist paths connecting the first states of these cycles.
\hfill \qed \end{proof}

We now obtain in the same way as before from \Cref{general:recognisability} and \Cref{twosnakes}, the desired certificate:

\begin{theorem}\label{thm:not-dww-cert}
An $\omega$-regular language $L$ is not in \textnormal{DWW} iff there exist five finite words $x, x_2, x_4 \in \Sigma^*$ and $x_1, x_3 \in \Sigma^+$, such that
\[x \cdot (x_1 + x_2 \cdot x_3^* \cdot x_4)^* \cdot x_1^\omega \subseteq L ~ \text{ and } ~ x \cdot (x_1 + x_2 \cdot x_3^* \cdot x_4)^* \cdot x_2 \cdot x_3^\omega \cap L = \emptyset.\]
\end{theorem}

Recall that DWW=DBW$\cap$DCW, so one would define a DWW certificate by disjuncting the certificates for DBW and DCW in \Cref{lem:not-dbw-cert,lem:not-dcw-cert}. \Cref{thm:not-dww-cert}, however, suggests a different certificate, and it is interesting to relate it to the ones for DBW and DCW.  Also note that while the DBW, DCW, and DPW certificates are covered by \cite[Lemma 14]{Wag79}, this is not the case for the DWW certificate in \Cref{thm:not-dww-cert}.

Recall that at the bottom of the depth hierarchy we have safety and co-safety languages, whose intersection is the set of bounded languages.

\begin{theorem}\label{thm:not-bounded}
An $\omega$-regular language $L$ is not a bounded language iff there exist six finite words $\hat{x}_{0}, \hat{x}_{1}, \hat{x}_2 \in \Sigma^*$ and $x_0, x_{1}, x_2 \in \Sigma^+$, such that
\[\hat{x}_0 \cdot x_0^* \cdot \hat{x}_{1} \cdot x_{1}^\omega \subseteq L \qquad \text{ and } \qquad \hat{x}_0 \cdot x_0^* \cdot \hat{x}_{2} \cdot x_{2}^\omega \cap L = \emptyset .\]
\end{theorem}

\begin{proof}
Assume first that $L$ is not bounded. Then, by \Cref{general:recognisability}, there exists a $\zug{\Lacc^{\text{bounded}}, \Llegal^{\text{bounded}}}$-refuter $\R$ for it. From this refuter we can extract three lassos: a $?$-labeled lasso from which we obtain $\hat{x}_{0}$ and $x_0$; a $\rej$-labeled lasso starting at the entry-point of the first lasso from which we obtain $\hat{x}_{1}$ and $x_1$; and a $\acc$-labeled lasso starting at the entry-point of the first lasso from which we obtain $\hat{x}_{2}$ and $x_2$.

For the other direction assume by way of contradiction that there is a deterministic $\zug{\Lacc^{\text{bounded}}, \Llegal^{\text{bounded}}}$-automaton $\A =\zug{\Sigma,Q,q_0,\delta, \tau, \gamma}$ with $\lang(\A)=L$. Assume that $\hat{x}_{0} \cdot x_0^\omega \in L$. Thus after reading $|Q|$ letters one state has been repeated and by the constraint it must be accepting. Thus $\hat{x}_{0} \cdot x_0^{|Q|} \cdot \hat{x}_{2} \cdot x_2^\omega \in L$ which is a contradiction. The other case is analogous.
\hfill \qed \end{proof}

\fi

\section{Discussion and Directions for Future Research}

The automation of decision procedures makes certification essential. We suggest to use the winning strategy of the refuter in expressiveness games as a certificate to inexpressibility. We show that beyond this {\em state-based certificate}, the strategy induces a {\em word-based certificate}, generated from words traversed along a ``flower structure" the strategy contains, as well as a {\em language-based certificate}, consisting of languages that under- and over-approximate the language in question and that are not separable by automata in the desired class.

While our work considers {\em expressive power}, one can use similar ideas in order to question the {\em size\/} of automata needed to recognize a given language. For example, in the case of a regular language $L$ of finite words, the Myhill-Nerode characterization \cite{Myh57,Ner58} suggests to refute the existence of deterministic finite word automata (DFW) with $n$ states for $L$ by providing $n+1$ prefixes that are not right-congruent. Using our approach, one can alternatively consider the winning strategy of Refuter in a game in which the set of annotations includes also the state space, and $\Llegal$ ensures consistency of the transition relation. Even more interesting is refutation of size in the setting of automata on infinite words. Indeed, there, minimization is NP-complete \cite{Sch10}, and there are interesting connections between polynomial certificates and possible membership in co-NP, as well as connections between size of certificates and succinctness of the different classes of automata.

Finally, while the approximation scheme we studied is based on suggested over- and under-approximating languages, it is interesting to study approximations that are based on more flexible distance measures \cite{DFT19,GGS17}.

\bibliographystyle{splncs04}
\bibliography{../ok}

\vfill

{\small\medskip\noindent{\bf Open Access} This chapter is licensed under the terms of the Creative Commons\break Attribution 4.0 International License (\url{http://creativecommons.org/licenses/by/4.0/}), which permits use, sharing, adaptation, distribution and reproduction in any medium or format, as long as you give appropriate credit to the original author(s) and the source, provide a link to the Creative Commons license and indicate if changes were made.}

{\small \spaceskip .28em plus .1em minus .1em The images or other third party material in this chapter are included in the chapter's Creative Commons license, unless indicated otherwise in a credit line to the material.~If material is not included in the chapter's Creative Commons license and your intended\break use is not permitted by statutory regulation or exceeds the permitted use, you will need to obtain permission directly from the copyright holder.}

\medskip\noindent\includegraphics{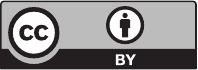}

\end{document}